\documentclass[draft,10pt,journal,a4paper,oneside,onecolumn]{IEEEtran}
\usepackage{amsmath,amsfonts, amssymb}

\newtheorem{thm}{Theorem}
\newtheorem{cor}[thm]{Corollary}
\newtheorem{lem}{Lemma}
\newtheorem{df}{Definition}
\newtheorem{rem}{Remark}
\newtheorem{prop}{Proposition}

\newcommand{\GFq}{\mathrm{GF}(q)}
\newcommand{\hA}{\widehat{A}}
\newcommand{\hg}{\widehat{g}}
\newcommand{\hcA}{\widehat{\mathcal{A}}}
\newcommand{\A}{\mathcal{A}}
\newcommand{\B}{\mathcal{B}}
\newcommand{\bhcA}{\boldsymbol{\widehat{\mathcal{A}}}}
\newcommand{\C}{\mathcal{C}}
\newcommand{\I}{\mathcal{I}}
\newcommand{\J}{\mathcal{J}}
\newcommand{\K}{\mathcal{K}}
\newcommand{\tK}{\widetilde{\mathcal K}}
\newcommand{\tk}{\widetilde{k}}
\newcommand{\M}{\mathcal{M}}
\newcommand{\R}{\mathcal{R}}
\newcommand{\cS}{\mathcal{S}}
\newcommand{\T}{\mathcal{T}}
\newcommand{\U}{\mathcal{U}}
\newcommand{\bU}{\overline{\mathcal{U}}}
\newcommand{\V}{\mathcal{V}}
\newcommand{\X}{\mathcal{X}}
\newcommand{\tcX}{\widetilde{\mathcal{X}}}
\newcommand{\tX}{\widetilde{X}}
\newcommand{\Y}{\mathcal{Y}}
\newcommand{\G}{\mathcal{G}}
\newcommand{\aalpha}{\boldsymbol{\alpha}}
\newcommand{\bbeta}{\boldsymbol{\beta}}
\newcommand{\kkappa}{\boldsymbol{\kappa}}
\newcommand{\haa}{\boldsymbol{\widehat a}}
\newcommand{\sfhaa}{\boldsymbol{\widehat{\mathsf a}}}
\newcommand{\ba}{\boldsymbol{a}}
\newcommand{\cc}{\boldsymbol{c}}
\newcommand{\ff}{\boldsymbol{f}}
\newcommand{\mm}{\boldsymbol{m}}
\newcommand{\bp}{\boldsymbol{p}}
\newcommand{\uu}{\boldsymbol{u}}
\newcommand{\vv}{\boldsymbol{v}}
\newcommand{\xx}{\boldsymbol{x}}
\newcommand{\yy}{\boldsymbol{y}}
\newcommand{\txx}{\widetilde{\boldsymbol{x}}}
\newcommand{\tx}{\widetilde{x}}
\newcommand{\XX}{\boldsymbol{X}}
\newcommand{\e}{\varepsilon}

\newcommand{\sfA}{\mathsf{A}}
\newcommand{\sfhA}{\widehat{\mathsf{A}}}
\newcommand{\sfaa}{\boldsymbol{\mathsf{a}}}
\newcommand{\sfcc}{\boldsymbol{\mathsf{c}}}
\newcommand{\Prod}{\operatornamewithlimits{\text{\Large $\times$}}}
\newcommand{\lrB}[1]{\left[{#1}\right]}
\newcommand{\lrb}[1]{\left\{{#1}\right\}}
\newcommand{\lrsb}[1]{\left({#1}\right)}
\newcommand{\lrbar}[1]{\left|{#1}\right|}
\newcommand{\co}{\mathrm{co}}
\newcommand{\cl}{\mathrm{cl}}
\newcommand{\Error}{\mathrm{Error}}
\newcommand{\zero}{\boldsymbol{0}}
\newcommand{\one}{\boldsymbol{1}}
\newcommand{\limn}{\lim_{n\to\infty}}
\newcommand{\Encoder}{\varphi}
\newcommand{\Decoder}{\psi}
\newcommand{\im}{\mathrm{Im}}
\newcommand{\bcA}{\boldsymbol{\mathcal{A}}}
\newcommand{\bchA}{\boldsymbol{\widehat{\mathcal{A}}}}
\newcommand{\bcAp}{\boldsymbol{\mathcal{A}}'}
\newcommand{\bpA}[1]{\bp_{\sfA_{#1}}}
\newcommand{\bpAp}[1]{\bp_{\sfA'_{#1}}}
\newcommand{\bphA}[1]{\bp_{\sfhA_{#1}}}
\newcommand{\pA}{p_{\sfA}}
\newcommand{\pAp}[1]{p_{\sfA'_{#1}}}
\newcommand{\phA}{p_{\sfhA}}
\newcommand{\alphaA}[1]{\alpha_{\sfA_{#1}}}
\newcommand{\betaA}[1]{\beta_{\sfA_{#1}}}
\newcommand{\alphaAp}[1]{\alpha_{\sfA'_{#1}}}
\newcommand{\betaAp}[1]{\beta_{\sfA'_{#1}}}
\newcommand{\alphahA}[1]{\alpha_{\sfhA_{#1}}}
\newcommand{\betahA}[1]{\beta_{\sfhA_{#1}}}
\newcommand{\aalphahA}[1]{\aalpha_{\sfhA_{#1}}}
\newcommand{\bbetahA}[1]{\bbeta_{\sfhA_{#1}}}
\newcommand{\aalphaA}{\aalpha_{\sfA}}
\newcommand{\bbetaA}{\bbeta_{\sfA}}
\newcommand{\aalphaAp}{\aalpha_{\sfA'}}
\newcommand{\bbetaAp}{\bbeta_{\sfA'}}

\allowdisplaybreaks

\title{
  Construction of Multiple Access Channel Codes Based on Hash Property
}
\author{
  Jun~Muramatsu
  and~Shigeki Miyake
  \thanks{J.~Muramatsu is with
    NTT Communication Science Laboratories, NTT Corporation,
    2-4, Hikaridai, Seika-cho, Soraku-gun, Kyoto 619-0237, Japan
    (E-mail: muramatsu.jun@lab.ntt.co.jp).
    S.~Miyake is with
    NTT Network Innovation Laboratories, NTT Corporation,
    Hikarinooka 1-1, Yokosuka-shi, Kanagawa 239-0847, Japan
    (E-mail: miyake.shigeki@lab.ntt.co.jp).
    This paper has been presented in part at \cite{ISIT2011b}
    and submitted to IEEE Transactions on Information Theory.
}}

\date{September 21, 2012}
\begin{document}
\maketitle

\begin{abstract}
  The aim of this paper is to introduce the construction of codes for a
  general discrete stationary memoryless multiple access channel based
  on the the notion of the hash property. Since an ensemble of sparse
  matrices has a hash property, we can use sparse matrices for code
  construction. Our approach has a potential advantage compared to the
  conventional random coding because it is expected that we can use some
  approximation algorithms by using the sparse structure of codes.
\end{abstract}

\section{Introduction}
This paper describes the construction of multiple access channel codes.
In a multiple access channel, two or more senders send messages to a
common receiver. The capacity region has been derived
in~\cite{A71}\cite{L72} for a scenario where two senders have different
private messages but no common message to be sent. This work has been
extended in \cite{SW73c} to a scenario where two senders have different
private messages and a common message to be sent. The capacity region
for two or more senders has been described
in~\cite[Section 15.3.5]{CT}\cite[Chapter 4]{EK10} in which there is no
common message. In \cite{H79}, the capacity region has been derived for
a general multiple access channel in which two or more senders have
messages common to some users. Applications of Low Density Parity Check
(LDPC) codes to a multiple access channel have been introduced
in~\cite{BD02}\cite{McE01}\cite{ITW03}. Furthermore, there are many
theoretical/experimental studies regarding the construction of multiple
access channel codes by using LDPC codes, e.g.~\cite{ADU02}\cite{SPS05}.
It should be noted that they assumed channel noises to be additive.

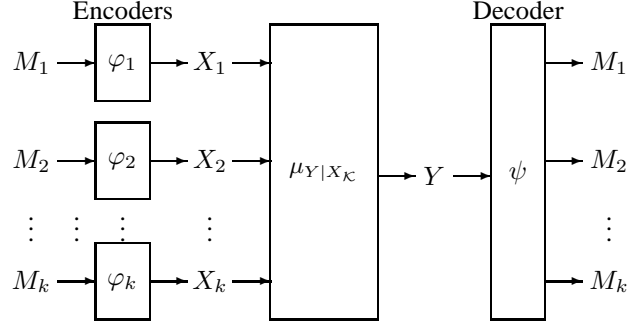
\begin{figure}[h]
  \begin{center}
    \unitlength 0.50mm
    \begin{picture}(157,90)(0,10)
      \put(3,78){\makebox(0,0){$M_1$}}
      \put(3,52){\makebox(0,0){$M_2$}}
      \put(3,36){\makebox(0,0){$\vdots$}}
      \put(3,20){\makebox(0,0){$M_{k}$}}
      \put(10,78){\vector(1,0){10}}
      \put(10,52){\vector(1,0){10}}
      \put(15,36){\makebox(0,0){$\vdots$}}
      \put(10,20){\vector(1,0){10}}
      \put(27,92){\makebox(0,0){Encoders}}
      \put(20,68){\framebox(14,20){$\Encoder_1$}}
      \put(20,42){\framebox(14,20){$\Encoder_2$}}
      \put(27,36){\makebox(0,0){$\vdots$}}
      \put(20,10){\framebox(14,20){$\Encoder_k$}}
      \put(34,78){\vector(1,0){10}}
      \put(50,78){\makebox(0,0){$X_1$}}
      \put(56,78){\vector(1,0){10}}
      \put(34,52){\vector(1,0){10}}
      \put(50,52){\makebox(0,0){$X_2$}}
      \put(56,52){\vector(1,0){10}}
      \put(50,36){\makebox(0,0){$\vdots$}}
      \put(34,20){\vector(1,0){10}}
      \put(50,20){\makebox(0,0){$X_k$}}
      \put(56,20){\vector(1,0){10}}
      \put(66,10){\framebox(28,78){\small $\mu_{Y|X_{\K}}$}}
      \put(94,48){\vector(1,0){10}}
      \put(109,48){\makebox(0,0){$Y$}}
      \put(114,48){\vector(1,0){10}}
      \put(131,92){\makebox(0,0){Decoder}}
      \put(124,10){\framebox(14,78){$\Decoder$}}
      \put(138,78){\vector(1,0){10}}
      \put(155,78){\makebox(0,0){$M_1$}}
      \put(138,52){\vector(1,0){10}}
      \put(155,52){\makebox(0,0){$M_2$}}
      \put(155,36){\makebox(0,0){$\vdots$}}
      \put(138,20){\vector(1,0){10}}
      \put(155,20){\makebox(0,0){$M_k$}}
    \end{picture}
  \end{center}
  \caption{Multiple Access Channel Coding: Private Messages}
  \label{fig:mac}
\end{figure}

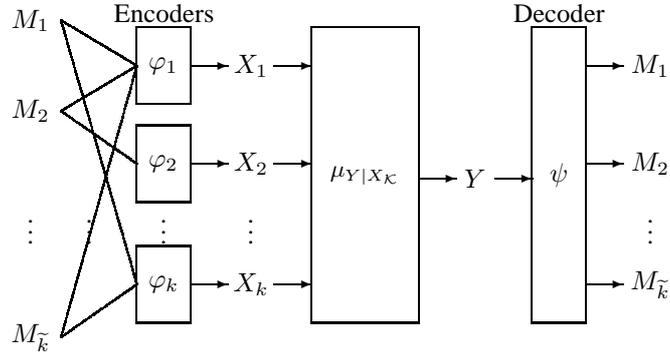
\begin{figure}[h]
  \begin{center}
    \unitlength 0.50mm
    \begin{picture}(157,95)(-5,5)
      \put(-8,90){\makebox(0,0){$M_1$}}
      \put(-8,66){\makebox(0,0){$M_2$}}
      \put(-8,36){\makebox(0,0){$\vdots$}}
      \put(-8,6){\makebox(0,0){$M_{\tk}$}}
      \qbezier(-0,90)(10,84)(20,78)
      \qbezier(-0,90)(10,55)(20,20)
      \qbezier(-0,66)(10,72)(20,78)
      \qbezier(-0,66)(10,59)(20,52)
      \qbezier(-0,6)(10,13)(20,20)
      \qbezier(-0,6)(10,42)(20,78)
      \put(7.5,36){\makebox(0,0){$\vdots$}}
      \put(27,92){\makebox(0,0){Encoders}}
      \put(20,68){\framebox(14,20){$\Encoder_1$}}
      \put(20,42){\framebox(14,20){$\Encoder_2$}}
      \put(27,36){\makebox(0,0){$\vdots$}}
      \put(20,10){\framebox(14,20){$\Encoder_k$}}
      \put(34,78){\vector(1,0){10}}
      \put(50,78){\makebox(0,0){$X_1$}}
      \put(56,78){\vector(1,0){10}}
      \put(34,52){\vector(1,0){10}}
      \put(50,52){\makebox(0,0){$X_2$}}
      \put(56,52){\vector(1,0){10}}
      \put(50,36){\makebox(0,0){$\vdots$}}
      \put(34,20){\vector(1,0){10}}
      \put(50,20){\makebox(0,0){$X_k$}}
      \put(56,20){\vector(1,0){10}}
      \put(66,10){\framebox(28,78){\small $\mu_{Y|X_{\K}}$}}
      \put(94,48){\vector(1,0){10}}
      \put(109,48){\makebox(0,0){$Y$}}
      \put(114,48){\vector(1,0){10}}
      \put(131,92){\makebox(0,0){Decoder}}
      \put(124,10){\framebox(14,78){$\Decoder$}}
      \put(138,78){\vector(1,0){10}}
      \put(155,78){\makebox(0,0){$M_1$}}
      \put(138,52){\vector(1,0){10}}
      \put(155,52){\makebox(0,0){$M_2$}}
      \put(155,36){\makebox(0,0){$\vdots$}}
      \put(138,20){\vector(1,0){10}}
      \put(155,20){\makebox(0,0){$M_{\tk}$}}
    \end{picture}
  \end{center}
  \caption{Multiple Access Channel Coding: Multiple Common Messages}
  \label{fig:han}
\end{figure}

\begin{figure}[h]
  \begin{center}
    \unitlength 0.50mm
    \begin{picture}(152,55)(0,10)
      \put(3,48){\makebox(0,0){$M_1$}}
      \put(8,48){\vector(1,0){12}}
      \put(3,16){\makebox(0,0){$M_2$}}
      \put(8,16){\vector(1,0){12}}
      \put(3,32){\makebox(0,0){$M_0$}}
      \put(13,40){\vector(1,0){7}}
      \put(13,24){\vector(1,0){7}}
      \put(8,32){\line(1,0){5}}
      \put(13,24){\line(0,1){16}}
      \put(27,60){\makebox(0,0){Encoders}}
      \put(20,34){\framebox(14,20){$\Encoder_1$}}
      \put(20,10){\framebox(14,20){$\Encoder_2$}}
      \put(34,44){\vector(1,0){10}}
      \put(49,44){\makebox(0,0){$X_1$}}
      \put(54,44){\vector(1,0){10}}
      \put(34,20){\vector(1,0){10}}
      \put(49,20){\makebox(0,0){$X_2$}}
      \put(54,20){\vector(1,0){10}}
      \put(64,10){\framebox(26,44){\small $\mu_{Y|X_1X_2}$}}
      \put(90,32){\vector(1,0){10}}
      \put(105,32){\makebox(0,0){$Y$}}
      \put(110,32){\vector(1,0){10}}
      \put(127,60){\makebox(0,0){Decoder}}
      \put(120,10){\framebox(14,44){$\Decoder$}}
      \put(134,48){\vector(1,0){10}}
      \put(149,48){\makebox(0,0){$M_1$}}
      \put(134,16){\vector(1,0){10}}
      \put(149,16){\makebox(0,0){$M_2$}}
      \put(134,32){\vector(1,0){10}}
      \put(149,32){\makebox(0,0){$M_0$}}
    \end{picture}
  \end{center}
  \caption{Two-user Multiple Access Channel Coding: Private and Common Messages}
  \label{fig:sw}
\end{figure}
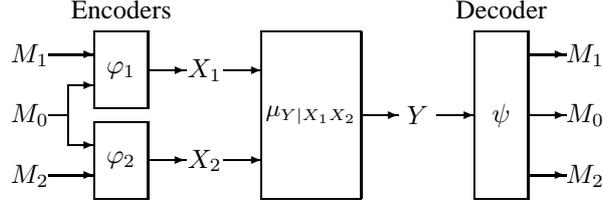

A contribution of this paper is to construct codes based on the notion
of the hash property~\cite{ISIT2010}\cite{HASH-BC}, which is a stronger
version of that introduced in \cite{HASH}\cite{HASH-UNIV}. Another
contribution is to construct codes by using sparse matrices for a
general discrete memoryless multiple access channel including asymmetric
one. We construct codes for the following scenarios:
\begin{itemize}
  \item Two or more senders have different private messages
  (Fig.\ref{fig:mac})
  \cite[Theorem 5 in Chapter 4]{EK10}\cite[Section 15.3.5]{CT},
  \item Two or more senders have messages common to senders
  (Fig.\ref{fig:han}) \cite{H79}, and
  \item Two senders have different private messages
  and a common message (Fig.\ref{fig:sw}) \cite{SW73c},
\end{itemize}
where additive channel noises are not assumed. It should be noted that
the first scenario includes two-sender scenario studied
in~\cite{A71}\cite{L72}. The last scenario is included in the second
scenario but we will discuss it separately because their code
constructions are different. The proof of all the theorems is based on
the notion of the hash property. It is an extension of the ensemble of
the random bin coding~\cite{C75}, the ensembles of linear
matrices~\cite{CSI82}, the universal class of hash functions~\cite{CW},
and the ensemble of sparse matrices~\cite{SWLDPC}. We use two lemmas,
one related to `saturation property\footnote{In~\cite{HASH}, they were
called `saturating property' and `collision-resistant property,'
respectively. We changed these terms following the suggestion of
Prof. T.S.~Han.}' (if the number of items is greater than the number of
bins then there is an assignment such that every bin contains at least
one item) and the other related to `collision-resistance
property$^{\text{1}}$' (if the number of bins is greater than the number
of items then there is an assignment such that every bin contains at
most one item) proved in \cite{HASH}\cite{HASH-BC}, where the lemma
related to the `collision-resistance property' is extended from a single
domain to multiple domains. They are reviewed in Section \ref{sec:hash}.
The saturation property is used to analyze the average encoding error
and the extended collision-resistance property is used to analyze the
average decoding error. It should be noted that the functions need not
be linear for the hash property but it is expected that the space and
time complexity of codes can be reduced compared with conventional
constructions by assuming the linearity of functions. In fact, we can
construct codes by using sparse matrices (with logarithmic column
degree) because an ensemble of sparse matrices has a hash
property~\cite{HASH-BC}. Although the implementation of encoding and
decoding operations of our approach is still intractable, our approach
has a potential advantage compared to the conventional random coding
presented in~\cite{CT}\cite{CSI82}\cite{EK10}\cite{H79}\cite{SW73c}
because it is expected\footnote{In fact, the direct application
of~\cite{FWK05}\cite{KFL01} does not provide good
performance. Implementation of our approach is left for a future
challenge.} that we can use some approximation algorithms such as the
sum-product algorithm~\cite{KFL01} and the linear programing
algorithm~\cite{FWK05} for encoding and decoding operations in the
proposed code with sparse matrix.

\section{Definitions and Notations}
\label{sec:def}

Throughout this paper, we use the following definitions and notations.
The cardinality of a set $\U$ is denoted by $|\U|$, $\U^c$ denotes the
complement of $\U$, and $\U\setminus\V\equiv\U\cap\V^c$ denotes the set
difference.

Column vectors and sequences are denoted in boldface. Let $A\uu$ denote
a value taken by a function $A:\U^n\to\bU$ at $\uu\in\U^n$, where $\U^n$
is the domain of the function and $\bU$ is the range of function. It
should be noted that $A$ may be nonlinear. When $A$ is a linear function
expressed by an $l\times n$ matrix, we assume that $\U\equiv\GFq$ is a
finite field and the range of functions is $\U^{l}$. For a set $\A$ of
functions, let $\im \A$ be defined as
\begin{align*}
  \im\A &\equiv \bigcup_{A\in\A}\{A\uu: \uu\in\U^n\}.
\end{align*}
We define a set $\C_A(\ba)$ as
\begin{align*}
  \C_A(\ba) &\equiv\{\uu: A\uu = \ba\}
\end{align*}
In the context of linear codes, $\C_A(\ba)$ is called a coset determined
by $\ba$. The random variables of a function $A$ and a vector
$\ba\in\im\A$ are denoted by the sans serif letters $\sfA$ and $\sfaa$,
respectively. It should be noted that some random variables are denoted
by the Roman letter (e.g. $M$, $U$, $V$, $X$, $Y$) which does not
represent a function.

For random variables $U$ and $V$, let $\mu_{U}$ be the probability
distribution and $\mu_{U|V}$ be the conditional probability
distribution. Then the entropy $H(U)$, the conditional entropy $H(U|V)$,
and the mutual information $I(U;V)$ are defined as
\begin{align*}
  H(U)
  &\equiv
  \sum_{u}\mu_U(u)\log\frac 1{\mu_U(u)}
  \\
  H(U|V)
  &\equiv
  \sum_{u,v}\mu_{U|V}(u|v)\mu_V(v)\log\frac 1{\mu_{U|V}(u|v)}
  \\
  I(U;V)
  &\equiv
  H(U)-H(U|V)
\end{align*}
where we assume a logarithm with base $2$ when the subscript of $\log$
is omitted. For random variables $U$, $V$, and $W$, let $I(U;V|W)$ be
the conditional mutual information defined as
\[
  I(U;V|W)\equiv H(U|W)-H(U|V,W).
\]

For $\uu\in\U^n$ and $\vv\in\V^n$, let $\nu_{\uu}$ and $\nu_{\uu|\vv}$
be the empirical distributions defined as
\begin{align}
  \nu_{\uu}(u)
  &\equiv
  \frac {|\{i\in\{1,\ldots,n\}: u_{i}=u\}|}n
  \label{eq:nu-uu}
  \\
  \nu_{\uu|\vv}(u|v)
  &\equiv \frac{\nu_{\uu\vv}(u,v)}{\nu_{\vv}(v)}
  \quad\text{for $v\in\V$ s.t. $\nu_{\vv}(v)>0$},
  \label{eq:nu-uu-vv}
\end{align}
where we use the relation
$\nu_{\uu|\vv}(u|v)\nu_{\vv}(v)=\nu_{\uu\vv}(u,v)$
even when $\nu_{\vv}(v)=0$. Let $p$ and $p'$ be probability
distributions on the same set $\U$ and let $q$ and $q'$ be conditional
probability distributions on the same set $\V$. Then divergence
$D(p\|p')$ and conditional divergence $D(q\|q'|p)$ are defined as
\begin{align*}
  D(p\parallel p')
  &\equiv
  \sum_{u\in\U}p(u)\log\frac{p(u)}{p'(u)}
  \\
  D(q\parallel q' | p)
  &\equiv
  \sum_{u\in\U} p(v)\sum_{v\in\V}q(v|u)\log\frac{q(v|u)}{q'(v|u)}.
\end{align*}

For the proof of theorems, we use the method of type developed
in~\cite{CK}, where we use the definition of a typical set introduced
in~\cite{HASH}\cite{UYE}. A set of typical sequences $\T_{U,\gamma}$ and
a set of conditionally typical sequences $\T_{U|V,\gamma}(\vv)$ are
defined as
\begin{align*}
  \T_{U,\gamma}
  &\equiv
  \lrb{\uu:
    D(\nu_{\uu}\|\mu_{U})<\gamma
  }
  \\
  \T_{U|V,\gamma}(\vv)
  &\equiv
  \lrb{\uu:
    D(\nu_{\uu|\vv}\|\mu_{U|V}|\nu_{\vv})<\gamma
  },
\end{align*}
respectively.
For $\uu\in\X^n$, $\vv\in\V^n$, the entropy $H(\uu)$, and the
conditional entropy $H(\uu|\vv)$ are defined as
\begin{align*}
  H(\uu)
  &\equiv
  \sum_{u}\nu_{\uu}(u)\log\frac 1{\nu_{\uu}(u)}
  \\
  H(\uu|\vv)
  &\equiv
  \sum_{u,v}\nu_{\uu|\vv}(u|v)\nu_{\vv}(v)\log\frac 1{\nu_{\uu|\vv}(u|v)},
\end{align*}
where $\nu_{\uu}$ and $\nu_{\uu|\vv}$ are defined as (\ref{eq:nu-uu})
and (\ref{eq:nu-uu-vv}), respectively. For $\gamma,\gamma'>0$, we define
\begin{align}
  \lambda_{\U}
  &\equiv \frac{|\U|\log(n+1)}n
  \label{eq:lambda}
  \\
  \iota_{\U}(\gamma)
  &\equiv
  -\sqrt{2\gamma}\log\frac{\sqrt{2\gamma}}{|\U|}
  \label{eq:iota}
  \\
  \iota_{\U|\V}(\gamma'|\gamma)
  &\equiv
  -\sqrt{2\gamma'}\log\frac{\sqrt{2\gamma'}}{|\U||\V|}
  +\sqrt{2\gamma}\log|\U|
  \label{eq:iotac}
  \\
  \eta_{\U}(\gamma)
  &\equiv
  -\sqrt{2\gamma}\log\frac{\sqrt{2\gamma}}{|\U|}
  +\frac{|\U|\log(n+1)}n
  \label{eq:def-eta}
  \\
  \eta_{\U|\V}(\gamma'|\gamma)
  &\equiv
  -\sqrt{2\gamma'}\log\frac{\sqrt{2\gamma'}}{|\U||\V|}
  +\sqrt{2\gamma}\log|\U|
  +\frac{|\U||\V|\log(n+1)}n,
  \label{eq:def-etac}
\end{align}
which comes from lemmas in Appendix~\ref{sec:type-theory} regarding the
method of types. It should be noted here that the product set
$\U\times\V$ is denoted by $\U\V$ when it appears in the subscript of
these functions.

For a mathematical statement $S$, we define $\chi(S)$ as
\begin{align*}
  \chi(S)
  &\equiv
  \begin{cases}
    1,&\text{if $S$ is true}
    \\
    0,&\text{if $S$ is false}.
  \end{cases}
\end{align*}

\section{Formal Description of Problems and Known Results}
\label{sec:review}

In this section, we review the problems of multiple access channel coding
and results regarding achievable regions. 

A multiple access channel has $k$ inputs and $1$ output. Let $\K$ be an
index set of the channel inputs, where $k\equiv|\K|$. Then the channel
is characterized by the conditional probability distribution
$\mu_{Y|X_{\K}}$, where $X_{\K}\equiv\{X_j\}_{j\in\K}$ is a $k$-tuple of
random variables corresponding to the inputs and $Y$ is a random
variable corresponding to the output. Let $\X_j$ be the alphabet of the
$j$-th channel input and $\Y$ be the alphabet of the channel output.

In the following, we review some coding scenarios that will be
discussed in subsequent sections. Let $\tK$ be an index set of messages
and $\tk\equiv|\tK|$. For each $i\in\tK$, let $\M_i$ be the alphabet of
the $i$-th message and $M_i$ be the random variable corresponding to the
$i$-th message, where we assume that the probability distribution of
$M_i$ is uniform on $\M_i$ for all $i\in\tK$. We also assume that random
variables $\{M_i\}_{i\in\tK}$ are mutually independent. Let
$p_{M_{\tK}}$ be the uniform distribution on $\M_{\tK}$. We use the
following notations:
\begin{align*}
  \Encoder_{\K}
  &\equiv
  \{\Encoder_j\}_{j\in\K}
  \\
  \M_{\tK}
  &\equiv\Prod_{i\in\tK}\M_i
  \\
  \mm_{\tK}
  &\equiv\lrb{\mm_i}_{i\in\tK},
  \quad\text{for given $\mm_i\in\M_i$, $i\in\tK$}
  \\
  R_{\tK}
  &\equiv\{R_i\}_{i\in\tK}.
\end{align*}
Let $\cl(\cdot)$ denote the closure of a region and  $\co(\cdot)$ denote
the closure of the convex hull of a region.

\subsection{Private Messages}

In this scenario, we assume that $\tK=\K$ and there are $k$ senders and
$k$ independent messages $M_{\K}$, where the $j$-th sender has access to
the $j$-th message $M_j$ and there is no common message. 

For a given block length $n$, a multiple access channel code
$(\Encoder_{\K},\Decoder)$ (Fig.\ref{fig:mac}) is defined by $k$
encoders $\Encoder_{\K}$ and one decoder $\Decoder$, where
\begin{align*}
  \Encoder_j:
  &\M_j\to\X_j^n\quad\text{for each}\ j\in\K
  \\
  \Decoder:
  & \Y^n\to\M_{\K}.
\end{align*}
Then the error probability of the code is defined as
\[
  \Error(\Encoder_{\K},\Decoder)
  \equiv
  \sum_{\substack{
      \mm_{\K}\in\M_{\K}\\
      \yy\in\Y^n
  }}
  \mu_{Y|X_{\K}}(\yy|\Encoder_{\K}(\mm_{\K}))
  p_{M_{\K}}(\mm_{\K})\chi(\Decoder(\yy)\neq\mm_{\K}).
\]
The rate $R_j$ of the $j$-th message is defined as
\[
  R_j\equiv\frac{\log|\M_j|}n
  \quad\text{for each}\ j\in\K.
\]
We call the rate vector $R_{\K}$ {\em achievable} if for all $\delta>0$
and all sufficiently large $n$, there is a code
$(\Encoder_{\K},\Decoder)$ with a rate vector $R_{\K}$ such that
\[
  \Error(\Encoder_{\K},\Decoder)<\delta.
\]

For a given $\{\mu_{X_j}\}_{j\in\K}$, let $\R(\{\mu_{X_j}\}_{j\in\K})$
be the set of all $k$-dimensional vectors $R_{\K}$ satisfying
\begin{equation}
  0\leq \sum_{j\in\J} R_j< I(X_{\J};Y|X_{\J^c})
  \quad\text{for all}\ \J\subset\K,
  \label{eq:sumRi}
\end{equation}
where the joint distribution $\mu_{X_{\K}Y}$ of random variable
$(X_{\K},Y)$ is given by
\begin{equation}
  \mu_{X_{\K}Y}(x_{\K},y)
  \equiv
  \mu_{Y|X_{\K}}(y|x_{\K})
  \lrB{\prod_{j\in\K}\mu_{X_j}(x_j)}.
  \label{eq:markov}
\end{equation}
For given $\mu_U$ and $\{\mu_{X_j|U}\}_{j\in\K}$, let
$\R(\mu_U,\{\mu_{X_j|U}\}_{j\in\K})$ be the set of all $k$-dimensional
vectors $R_{\K}$ satisfying
\begin{equation}
  0\leq \sum_{j\in\J} R_j< I(X_{\J};Y|U,X_{\J^c})
  \quad\text{for all}\ \J\subset\K,
  \label{eq:sumRi-ts}
\end{equation}
where the joint distribution $\mu_{UX_{\K}Y}$ of random variable
$(U,X_{\K},Y)$ is given by
\begin{equation}
  \mu_{UX_{\K}Y}(u,x_{\K},y)
  \equiv
  \mu_{Y|X_{\K}}(y|x_{\K})
  \lrB{\prod_{j\in\K}\mu_{X_j|U}(x_j|u)}\mu_{U}(u).
  \label{eq:markov-ts}
\end{equation}
Then the achievable region for this scenario is given as described below.
\begin{prop}[{\cite[Theorem 15.3.6]{CT}\cite[Theorem 4.5]{EK10}}]
\label{prop:region}
The achievable region for this scenario is given as
\begin{equation}
  \co\lrsb{\bigcup_{\{\mu_{X_j}\}_{j\in\K}}\R(\{\mu_{X_j}\}_{j\in\K})},
  \label{eq:region}
\end{equation}
which is equivalent to
\begin{equation}
  \bigcup_{\mu_U,\{\mu_{X_j}\}_{j\in\K}}
  \cl\lrsb{\R(\mu_U,\{\mu_{X_j|U}\}_{j\in\K})},
  \label{eq:region-ts}
\end{equation}
where $|\U|\leq k$.
\end{prop}
\begin{rem}
It should be noted that this proposition includes the result
of~\cite{A71}\cite{L72} corresponding to the case of two encoders.
The equivalence  of the two regions (\ref{eq:region}) and
(\ref{eq:region-ts}) can be shown from \cite[Theorem 15.3.6]{CT}
and~\cite[Theorem 4.5]{EK10} by considering the operational definition
of capacity region.
\end{rem}

In Section~\ref{sec:mac-ts}, for a given
$R_{\K}\in\R(\mu_U,\{\mu_{X_j|U}\}_{j\in\K})$, we construct a code with
the rate vector $R_{\K}$ based on the coded time sharing technique.
It should be noted that we can construct a code with a rate vector
$R_{\K}\in\R(\{\mu_{X_j}\}_{j\in\K})$ by letting $U$ be a constant, that
is, $|\U|=1$. In fact,
$\R(\{\mu_{X_j}\}_{j\in\K})=\R(\mu_U,\{\mu_{X_j|U}\}_{j\in\K})$
when $U$ is a constant. The achievability of the region
(\ref{eq:region}) with a proposed code can be proved by using the
time-sharing argument.

\subsection{Multiple Common Messages}

In this scenario, we assume that there are $\tk$ independent messages
$M_{\tK}$ and $k$ encoders, where the $j$-th encoder has access to the
messages $M_{\tK_j}\equiv\{M_i\}_{i\in\tK_j}$ specified by
$\tK_j\subset\tK$ for each $j\in\K$.

For a given block length $n$, a multiple access channel code
$(\Encoder_{\K},\Decoder)$ (Fig.\ref{fig:han}) is defined by
$k$ encoders $\Encoder_{\K}\equiv\{\Encoder_j\}_{j\in\K}$ and one
decoder $\Decoder$, where
\begin{align*}
  \Encoder_j: &\M_{\tK_j}\to\X_j^n\quad\text{for each}\ j\in\K
  \\
  \Decoder: &\Y^n\to\M_{\K}.
\end{align*}
Then the error probability of the code is defined by
\[
  \Error(\Encoder_{\K},\Decoder)
  \equiv
  \sum_{\substack{
      \mm_{\tK}\in\M_{\tK}\\
      \yy\in\Y^n
  }}
  \mu_{Y|X_{\K}}(\yy|\Encoder_{\K}(\mm_{\tK}))
  p_{M_{\tk}}(\mm_{\tK})
  \chi(\Decoder(\yy)\neq\mm_{\tK}).
\]
For each $i\in\tK$, the rate $R_i$ of the $i$-th message is defined by
\[
  R_i\equiv\frac{\log|\M_i|}n.
\]
We call the rate vector $R_{\tK}$ {\em achievable} if for all $\delta>0$
and all sufficiently large $n$, there is a code
$(\Encoder_{\K},\Decoder)$ with a rate vector $R_{\tK}$ such that
\[
  \Error(\Encoder_{\K},\Decoder)<\delta.
\]

For each $i\in\tK$, let $\tX_i$ be an auxiliary random variable and
$\tcX_i$ is the alphabet of $\tX_i$. For a given
$\{\mu_{\tX_i}\}_{i\in\tK}$ and a set $\{f_j\}_{j\in\K}$ of functions
\[
  f_j:\tcX_{\tK_j}\to\X_j
  \quad\text{for each}\ j\in\K,
\]
let\footnote{The subscript `H' comes from the author Han of~\cite{H79}.}
$\R_{\mathrm H}(\{\mu_{\tX_i}\}_{i\in\tK},\{f_j\}_{j\in\K})$ be the set
of all $s$-dimensional vectors $R_{\tK}$ satisfying
\begin{equation}
  0\leq \sum_{i\in\I} R_i
  < I(\tX_{\I};Y|\tX_{\I^c})
  \quad\text{for all}\ \I\subset\tK,
  \label{eq:sumRi-han}
\end{equation}
where the joint distribution  $\mu_{\tX_{\tK}X_{\K}Y}$ of random
variable $(\tX_{\tK},X_{\K},Y)$ is  given by
\[
  \mu_{\tX_{\tK}X_{\K}Y}(\tx_{\tK},x_{\K},y)
  \equiv
  \mu_{Y|X_{\K}}(y|x_{\K})
  \lrB{\prod_{j\in\K}\chi(f_j(\tx_{\tK_j})=x_j)}
  \lrB{\prod_{i\in\tK}\mu_{\tX_i}(\tx_i)}.
\]
Then the achievable region for this scenario is given as described below.
\begin{prop}[{\cite[Theorem 4.1]{H79}}]
The achievable region for this scenario is given as
\begin{equation}
  \co\lrsb{
    \bigcup_{\{\mu_{\tX_i}\}_{i\in\tK},\{f_j\}_{j\in\K}}
    \R_{\mathrm H}(\{\mu_{\tX_i}\}_{i\in\tK},\{f_j\}_{j\in\K})
  },
  \label{eq:region-han}
\end{equation}
where 
\[
  |\tcX_i|\leq
  |\tK|+\prod_{\substack{
      j\in\K:\\
      i\in\tK_j
  }}
  |\X_j|
  \quad\text{for all}\ i\in\tK.
\]
\end{prop}

In Section \ref{sec:han}, for a given
$\R_{\mathrm H}(\{\mu_{\tX_i}\}_{i\in\tK},\{f_j\}_{j\in\K})$, we
construct a code with the rate vector $R_{\tK}$. The achievability of
region (\ref{eq:region-han}) with a proposed code can be proved by using
the time-sharing argument.

In the following, let us consider a scenario (Fig.\ref{fig:sw}) in which
one of two senders has access to messages $M_0$ and $M_1$ and another
sender has access to messages $M_0$ and $M_2$, where $M_0$ denotes a
common message. It is a special case of the above scenario, where
$\K\equiv\{1,2\}$, $\tK\equiv\{0,1,2\}$, $\tK_1\equiv\{0,1\}$, and
$\tK_2\equiv\{0,2\}$.

Let $R_0$ be the encoding rate of the common message and $R_1$ and $R_2$
be the encoding rate of the private message of the respective encoders.
Let\footnote{The subscript `SW' comes from the authors Slepian and Wolf
of~\cite{SW73c}.} $\R_{\mathrm SW}(\mu_{X_0},\mu_{X_1|X_0},\mu_{X_2|X_0})$
be the set of all $(R_0,R_1,R_2)$ satisfying
\begin{gather}
  R_0\geq 0
  \label{eq:R0-sw}
  \\
  0\leq R_1 < I(X_1;Y|X_0,X_2)
  \label{eq:R1-sw}
  \\
  0\leq R_2 < I(X_2;Y|X_0,X_1)
  \label{eq:R2-sw}
  \\
  R_1+R_2 < I(X_1,X_2;Y|X_0)
  \label{eq:R1R2-sw}
  \\
  R_0+R_1+R_2 < I(X_1,X_2;Y)
  \label{eq:R0R1R2-sw}
\end{gather}
where the joint distribution  $\mu_{X_0X_1X_2Y}$
of random variables $(X_0,X_1,X_2,Y)$ is  given by
\begin{align}
  \mu_{X_0X_1X_2Y}(x_0,x_1,x_2,y)
  &\equiv
  \mu_{Y|X_1X_2}(y|x_1,x_2)\mu_{X_1|X_0}(x_1|x_0)
  \mu_{X_2|X_0}(x_2|x_0)
  \mu_{X_0}(x_0).
  \label{eq:markov-sw}
\end{align}
It should be noted that (\ref{eq:markov-sw}) implies the fact that
the right hand side of (\ref{eq:R0R1R2-sw}) is equal to
$I(X_1X_2X_0;Y)$. Then, the rate region is given as described below.
\begin{prop}[\cite{SW73c}]
For the scenario in which two receivers have access to their private
message and a common message, the achievable region is given as
\begin{equation}
  \co\lrsb{
    \bigcup_{\mu_{X_0},\mu_{X_1|X_0},\mu_{X_2|X_0}}
    \R_{\mathrm SW}(\mu_{X_0},\mu_{X_1|X_0},\mu_{X_2|X_0})
  },
  \label{eq:region-sw}
\end{equation}
where 
\[
  |\X_0|\leq\min\{|\Y|+3,|\X_1||\X_2|+2\}.
\]
\end{prop}

\begin{rem}
It should be noted that region (\ref{eq:region-sw}) is equivalent to
the region obtained from (\ref{eq:region-han}). This has been proven
in~\cite{H79}.
\end{rem}

In Section \ref{sec:sw}, for given
$(R_0,R_1,R_2)\in\R_{\mathrm SW}(\mu_{X_0},\mu_{X_1|X_0},\mu_{X_2|X_0})$,
we construct a code with  the rate vector $(R_0,R_1,R_2)$.
The construction is a typical example of the superposition coding introduced
in~\cite{SW73c} based on the hash property, and it is different from the
construction presented in Section \ref{sec:han}. The achievability of
region (\ref{eq:region-sw}) with a proposed code can be proved by using
the time-sharing argument.

\section{$(\aalpha,\bbeta)$-hash property}
\label{sec:hash}

In this section, we introduce the hash property first introduced
in~\cite{ISIT2010}\cite{HASH-BC} and its implications. This notion is
used for the proof of theorems.

\subsection{Formal Definition}
Here, we introduce the hash property for an ensemble of functions. It
has been introduced in \cite{ISIT2010}\cite{HASH-BC} and requires
stronger conditions than those introduced in \cite{HASH}.
\begin{df}[\cite{HASH-BC}\cite{ISIT2010}]
Let $\bcA\equiv\{\A^{(n)}\}_{n=1}^{\infty}$ be a sequence of sets such
that $\A^{(n)}$ is a set of functions $A:\U^n\to\im\A^{(n)}$. For a
probability distribution $p_{\sfA,n}$ on $\A^{(n)}$, we call a sequence
$(\bcA,\bpA{})\equiv\{(\A^{(n)},p_{\sfA,n})\}_{n=1}^{\infty}$
an {\em ensemble}. Then, $(\bcA,\bpA{})$ has a
$(\aalphaA{},\bbetaA{})$-{\em hash
 property}\footnote{In~\cite{HASH-BC}\cite{ISIT2010}\cite{ISIT2011a}\cite{ISIT2011b},
it is called the `strong hash property.' Throughout this paper, we call
it simply the `hash property.'} (or simply {\em hash property}) if
there are two sequences $\aalphaA\equiv\{\alphaA{}(n)\}_{n=1}^{\infty}$
and $\bbetaA\equiv\{\betaA{}(n)\}_{n=1}^{\infty}$, which depend on
$\{p_{\sfA,n}\}_{n=1}^{\infty}$, such that
\begin{align}
  &\limn \alphaA{}(n)=1
  \tag{H1}
  \label{eq:alpha}
  \\
  &\limn \betaA{}(n)=0
  \tag{H2}
  \label{eq:beta}
\end{align}
and
\begin{align}
  \sum_{\substack{
      \uu'\in\U^n\setminus\{\uu\}:
      \\
      p_{\sfA,n}(\{A: A\uu = A\uu'\})>\frac{\alphaA{}(n)}{|\im\A_n|}
  }}
  p_{\sfA,n}\lrsb{\lrb{A: A\uu = A\uu'}}
  \leq
  \betaA{}(n)
  \tag{H3}
  \label{eq:hash}
\end{align}
for any $n$ and $\uu\in\U^n$. Throughout this paper, we omit the
dependence of $\A$, $\pA$, $\alphaA{}$ and $\betaA{}$ on $n$.
\end{df}
\begin{rem}
In~\cite{HASH}\cite{ISIT2010}, an ensemble is required to satisfy the
condition
\[
  \limn \frac1n\log\frac{|\bU^{(n)}|}{|\im\A^{(n)}|}=0,
\]
where $\bU^{(n)}$ is the range of functions. This condition is omitted
because it is unnecessary for the results reported in this paper.
\end{rem}

Let us remark on the condition (\ref{eq:hash}). This condition requires
the sum of the collision probabilities
$p_{\sfA}\lrsb{\lrb{A: A\uu = A\uu'}}$, which is greater than
$\alphaA{}/|\im\A|$, to be bounded by $\betaA{}$, where the sum is taken
over all $\uu'$ except $\uu$. For an ensemble of sparse matrices,
$\alphaA{}$ represents the difference between $(\bcA,\bpA{})$ and the
ensemble of all linear matrices with uniform distribution, and
$\betaA{}$ represents the upper bound of the probability that the set
$\{\uu\in\U^n:A\uu=0\}$, which is called a code in the context of linear
codes, has low weight vectors. It should be noted that this condition
implies
\begin{equation}
  \sum_{\substack{
      \uu\in\T
      \\
      \uu'\in\T'
  }}
  \pA\lrsb{\lrb{A: A\uu = A\uu'}}
  \leq
  |\T\cap\T'|
  +
  \frac{|\T||\T'|\alphaA{}}{|\im\A|}
  +
  \min\{|\T|,|\T'|\}\betaA{}
  \tag{H3'}
  \label{eq:whash}
\end{equation}
for any $\T,\T'\subset\U^n$, which is introduced in \cite{HASH}.
The stronger condition (\ref{eq:hash}) is required for
Lemmas~\ref{lem:hash-AB} and~\ref{lem:multi-CRP}, which will appear
later.

It should be noted that when $\A$ is a two-universal class of hash
functions \cite{CW} and  $\pA$ is the uniform distribution on $\A$,
then $(\bcA,\bpA{})$ has a $(\one,\zero)$-hash property, where random
bin coding \cite{C75} and the set of all linear functions \cite{CSI82}
are examples of the two-universal class of hash functions.
It is proved in \cite[Section III-B]{HASH-BC} that an ensemble of sparse
matrices has a hash property. From this fact, this ensemble of sparse
matrices can be applied to all results in this paper. This implies that
all proposed codes can be constructed by using sparse matrices.

We have the following lemma, where it is unnecessary to assume the
linearity of functions assumed in \cite{HASH}\cite{HASH-UNIV}. It is one
of the advantages of introducing a stronger version of the hash property.

\begin{lem}[{\cite[Lemma 4]{HASH-BC}}]
\label{lem:hash-AB}
Let $(\bcA,\bpA{})$ and $(\bcAp,\bpAp{})$ be ensembles satisfying a
$(\aalphaA,\bbetaA)$-hash property and a $(\aalphaAp{},\bbetaAp{})$-hash
property, respectively. Let $\A\in\bcA$ (resp. $\A'\in\bcAp$) be a set
of functions $A:\U^n\to\im\A$ (resp. $A':\U^n\to\im\A'$). Let
$\hcA\equiv\A\times\A'$  and $\hA\equiv(A,A')\in\hcA$ defined as
\[
  \hA\uu\equiv(A\uu,A'\uu)
  \quad\text{for each $\hA\in\hcA$, $\uu\in\U^n$}.
\]
Let $\phA$  be a joint distribution on $\hcA$ defined as
\[
  p_{\sfhA}(A,A')\equiv \pA(A)\pAp{}(A').
\]
Then the ensemble $(\bchA,\bphA{})$ has a
$(\aalphahA{},\bbetahA{})$-hash property, where
$(\alphahA{},\betahA{})$ is defined as
\begin{align*}
  \alphahA{}&\equiv \alphaA{}\alphaAp{}
  \\
  \betahA{}&\equiv \betaA{}+\betaAp{}.
\end{align*}
\end{lem}

\subsection{Two Implications of Hash Property}

We review two implications of the hash property, which is introduced in
\cite{HASH}. These two implications connect the number of bins and
messages (items) and are derived from the hash property by adjusting the
number of bins taking account of the number of sequences.

In the following, let $\A$ be a set of functions $A:\U^n\to\im\A$,
where an item is a member of $\U^n$. A function $A$ assigns a label
$A\uu$ of bin to an item $\uu\in\U^n$.

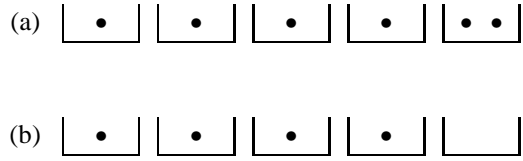
\begin{figure}[t]
  \begin{center}
    \unitlength = 0.5mm
    \begin{picture}(140,30)(-20,5)
      \put(-10,15){\makebox(0,0){(a)}}
      \put(0,10){\line(1,0){20}}
      \put(0,10){\line(0,1){10}}
      \put(20,10){\line(0,1){10}}
      \put(25,10){\line(1,0){20}}
      \put(25,10){\line(0,1){10}}
      \put(45,10){\line(0,1){10}}
      \put(50,10){\line(1,0){20}}
      \put(50,10){\line(0,1){10}}
      \put(70,10){\line(0,1){10}}
      \put(75,10){\line(1,0){20}}
      \put(75,10){\line(0,1){10}}
      \put(95,10){\line(0,1){10}}
      \put(100,10){\line(1,0){20}}
      \put(100,10){\line(0,1){10}}
      \put(120,10){\line(0,1){10}}
      \put(10,15){\makebox(0,0){{\small $\bullet$}}}
      \put(35,15){\makebox(0,0){{\small $\bullet$}}}
      \put(60,15){\makebox(0,0){{\small $\bullet$}}}
      \put(85,15){\makebox(0,0){{\small $\bullet$}}}
      \put(106,15){\makebox(0,0){{\small $\bullet$}}}
      \put(114,15){\makebox(0,0){{\small $\bullet$}}}
    \end{picture}
    \\
    \begin{picture}(140,30)(-20,5)
      \put(-10,15){\makebox(0,0){(b)}}
      \put(0,10){\line(1,0){20}}
      \put(0,10){\line(0,1){10}}
      \put(20,10){\line(0,1){10}}
      \put(25,10){\line(1,0){20}}
      \put(25,10){\line(0,1){10}}
      \put(45,10){\line(0,1){10}}
      \put(50,10){\line(1,0){20}}
      \put(50,10){\line(0,1){10}}
      \put(70,10){\line(0,1){10}}
      \put(75,10){\line(1,0){20}}
      \put(75,10){\line(0,1){10}}
      \put(95,10){\line(0,1){10}}
      \put(100,10){\line(1,0){20}}
      \put(100,10){\line(0,1){10}}
      \put(120,10){\line(0,1){10}}
      \put(10,15){\makebox(0,0){{\small $\bullet$}}}
      \put(35,15){\makebox(0,0){{\small $\bullet$}}}
      \put(60,15){\makebox(0,0){{\small $\bullet$}}}
      \put(85,15){\makebox(0,0){{\small $\bullet$}}}
    \end{picture}
  \end{center}
  \caption{
    Properties connecting the number of bins and items (black dots, messages).
    (a) Saturation property: every bin contains at least one item.
    (b) Collision-resistance property: every bin contains at most one item.
  }
  \label{fig:principle}
\end{figure}

\noindent{\bf Saturation property:}
We prepare a method that finds a typical sequence for each bin.
The saturation property is a characteristic of the hash property.
Figure \ref{fig:principle} (a) represents the ideal situation of this
property. When the number of bins is smaller than the number of black
dots, we can find a suitable function whereby every bin has at least one
black dot. This is because the hash property tends to avoid collisions.
It should be noted that it is sufficient for coding problems to satisfy
this property for `almost all (close to probability one)' bins by
letting the ratio
$[\text{the number of bins}]/[\text{the number of black dots}]$
be close to zero. To find a typical sequence from each bin, we use the
minimum-divergence operation introduced in the construction of codes,
where this operation finds a typical sequence when there is. In this
situation, the black dots correspond to typical sequences.

We have the following lemma, which is related to the saturation property.
\begin{lem}[{\cite[Lemma 2]{HASH}}]
\label{lem:SP}
Assume that the distribution of a random variable $\sfaa$ is uniform on
$\im\A$ and $\sfaa$ and $\sfA$ are mutually independent. If $(\A,\pA)$
satisfies (\ref{eq:whash}), then
\[
  p_{\sfA\sfaa}\lrsb{\lrb{(A,\ba):
      \T\cap\C_{A}(\ba)=\emptyset
  }}
  \leq
  \alphaA{}-1+\frac{|\im\A|\lrB{\betaA{}+1}}{|\T|}
\]
for any $\T\subset\U^n$.
\end{lem}

We prove the saturation property from Lemma~\ref{lem:SP}. We have
\begin{align}
  E_{\sfA}\lrB{
    p_{\sfcc}\lrsb{\lrb{\cc:  \T\cap\C_{\sfA}(\cc)=\emptyset}}
  }
  &=
  p_{\sfA\sfcc}\lrsb{\lrb{(A,\cc): \T\cap\C_A(\cc)=\emptyset}}
  \notag
  \\
  &\leq
  \alphaA{}-1+\frac{|\im\A|\lrB{\betaA{}+1}}{|\T|}.
\end{align}
By assuming that $|\im\A|/|\T|$ vanishes as $n\to\infty$, we have the
fact that there is a function $A$ such that
\[
  p_{\sfcc}\lrsb{\lrb{\cc: \T\cap\C_A(\cc)=\emptyset}}
  <\delta
\]
for any $\delta>0$ and sufficiently large $n$. Since the relation
$\T\cap\C_A(\cc)=\emptyset$ corresponds to an event where there is no
$\uu\in\T$ in bin $\C_A(\cc)$, we have the fact that we can find a
member of $\T$ in a randomly selected bin with probability close to one.

\noindent{\bf Collision-resistance property:}
A good code assigns a message to a codeword that is different from the
codewords of other messages, where the error probability is as small as
possible. The collision-resistance property is another characteristic of
the hash property. Figure \ref{fig:principle} (b) shows the ideal
situation as regards this property, where the black dots represent
messages we want to distinguish. When the number of bins is greater than
the number of black dots, we can find a good function that allocates the
black dots to the different bins. This is because the hash property
tends to avoid the collision. It should be noted that it is sufficient
for coding problems to satisfy this property for `almost all
(close to probability one)' black dots by letting the ratio
$[\text{the number of black dots}]/[\text{the number of bins}]$
be close to zero. This property is used to estimate the decoding error
probability. In this situation, the black dots correspond to typical
sequences.

We have the following lemma, which is related to to the
collision-resistance property.
\begin{lem}[{\cite[Lemma 1]{HASH}}]
\label{lem:CRP}
If $(\A,\pA)$ satisfies (\ref{eq:whash}), then
\[
  \pA\lrsb{\lrb{
      A: \lrB{\G\setminus\{\uu\}}\cap\C_A(A\uu)\neq \emptyset
  }}
  \leq 
  \frac{|\G|\alphaA{}}{|\im\A|} + \betaA{}.
\]
for all $\G\subset\U^n$  and $\uu\in\U^n$.
\end{lem}

We prove the collision-resistance property from Lemma~\ref{lem:CRP}.
Let $\mu_U$ be the probability distribution on $\G\subset\U^n$. We have
\begin{align}
  E_{\sfA}\lrB{
    \mu_U\lrsb{\lrb{\uu:
	\lrB{\G\setminus\{\uu\}}\cap\C_{\sfA}(\sfA\uu)\neq\emptyset}
    }
  }
  &
  \leq
  \sum_{\uu\in\G}\mu_U(\uu)
  p_{\sfA}\lrsb{\lrb{A: \lrB{\G\setminus\{\uu\}}\cap\C_A(A\uu)\neq\emptyset}}
  \notag
  \\
  &
  \leq
  \sum_{\uu\in\G}\mu_U(\uu)
  \lrB{\frac{|\G|\alphaA{}}{|\im\A|} + \betaA{}}
  \notag
  \\
  &
  \leq
  \frac{|\G|\alphaA{}}{|\im\A|} + \beta_A{}.
\end{align}
By assuming that $|\G|/|\im\A|$ vanishes as $n\to\infty$, we have the
fact that there is a function $A$ such that
\[
  \mu_U\lrsb{\lrb{\uu:
      \lrB{\G\setminus\{\uu\}}\cap\C_A(A\uu)\neq\emptyset}}
  <\delta
\]
for any $\delta>0$ and sufficiently large $n$. Since the relation
$\lrB{\G\setminus\{\uu\}}\cap\C_A(A\uu)\neq\emptyset$ corresponds to an
event where there is $\uu'\in\G$ such that $\uu$ and $\uu'$ are
different members of the same bin (they have the same codeword
determined by $A$), we have the fact that the members of $\G$ are
located in different bins (the members of $\G$ can be decoded correctly)
with probability close to one.

\subsection{Channel Coding Based on Hash Property}
Now, we explain an intuitive construction of a channel code in terms of
the saturation property and the collision-resistance property, where the
construction is introduced in~\cite{HASH}.

\begin{figure}[t]
  \begin{center}
    \unitlength 0.4mm
    \begin{picture}(176,70)(0,0)
      \put(82,60){\makebox(0,0){Encoder}}
      \put(65,35){\makebox(0,0){$\cc$}}
      \put(70,35){\vector(1,0){10}}
      \put(30,17){\makebox(0,0){$\mm$}}
      \put(45,17){\vector(1,0){35}}
      \put(80,10){\framebox(18,32){$\hg_{AB}$}}
      \put(98,26){\vector(1,0){20}}
      \put(128,26){\makebox(0,0){$\xx$}}
      \put(55,0){\framebox(54,52){}}
    \end{picture}
    \\
    \begin{picture}(176,70)(0,0)
      \put(82,60){\makebox(0,0){Decoder}}
      \put(30,35){\makebox(0,0){$\cc$}}
      \put(35,35){\vector(1,0){10}}
      \put(0,17){\makebox(0,0){$\yy$}}
      \put(10,17){\vector(1,0){35}}
      \put(45,10){\framebox(18,32){$\hg_A$}}
      \put(63,26){\vector(1,0){10}}
      \put(83,26){\makebox(0,0){$\xx$}}
      \put(93,26){\vector(1,0){10}}
      \put(103,19){\framebox(30,14){$B$}}
      \put(133,26){\vector(1,0){20}}
      \put(167,26){\makebox(0,0){$\mm$}}
      \put(20,0){\framebox(124,52){}}
    \end{picture}
  \end{center}
  \caption{Construction of Channel Code}
  \label{fig:channel-code}
\end{figure}
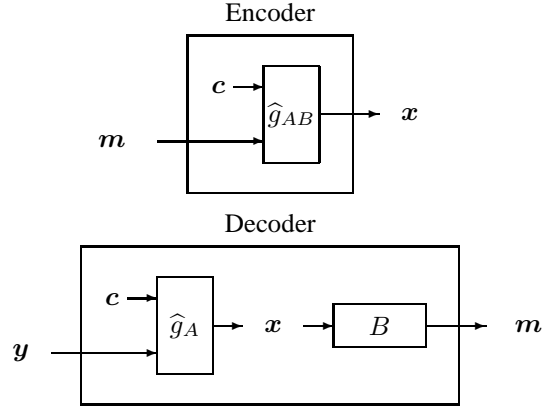

We prepare two functions $A:\X^n\to\im\A$, $B:\X^n\to\im\B$, and a
vector $\cc\in\im\A$, and assume that they are shared by an encoder and
a decoder. It should be noted that $|\im\A|$ (resp. $|\im\B|$) is the
number of bins specified by $A$ (resp. $B$). The function $A$ is
analogous to a parity check matrix in the context of linear codes. The
set $\im\B$ is the set of all messages and $|\im\B|$ is equal to the
number of messages.

The flow of vectors is is illustrated in Fig.~\ref{fig:channel-code}.
Let $\mm\in\im\B$ be a message, $\xx\in\X^n$ be a channel input, and
$\yy\in\Y^n$ be a channel output. For $\cc$ and a message $\mm$,
a function $\hg_{AB}$ generates a typical sequence $\xx\in\T_{X,\gamma}$
as a channel input, where $A\xx=\cc$ and $B\xx=\mm$ are satisfied.
The decoder  reproduces the channel input $\xx$ by using $\hg_{A}$ from
$\cc$ and a channel output $\yy$. Since $(\xx,\yy)$ is jointly typical
and $B\xx=\mm$, the decoding succeeds if the amount of information of
$\cc$ is greater than $H(X|Y)$ to satisfy the collision-resistance
property. In fact, there are about $2^{n H(X|Y)}$ conditional typical
sequences $\xx$ for given $\yy$ and it is sufficient to prepare bins
specified by $A$ more than $2^{nH(X|Y)}$ to distinguish conditional
typical sequences. Formally, this condition corresponds to
\[
  \frac{\log|\im\A|}n>H(X|Y).
\]
On the other hand, the total rate of $\cc$ and $\mm$ should be less than
$H(X)$ to satisfy the saturating property. Since there is at most
$2^{nH(X)}$ typical sequences, it is sufficient to prepare bins
specified by $A$ and $B$ less than  $2^{nH(X)}$. Formally, this
condition corresponds to
\[
  \frac{\log|\im\A||\im\B|}n<H(X).
\]
Since it is sufficient to satisfy these two inequalities, we have the
fact that there is a code when
\[
  \frac{\log|\im\B|}n<H(X)-H(X|Y).
\]
This implies that we can set the encoding rate of messages close to
$H(X)-H(X|Y)=I(X;Y)$.

In this paper, we extend this approach to construct a multiple access
channel code.

\subsection{Multiple Extension of Collision Resistance Property}
To prove the achievability of a multiple access channel code based on
the hash property, we extend the lemma related to the
collision-resistance property. The hash property is needed to prove the
following lemma, and this is another reason why the hash property is
introduced. We use the following notations:
\begin{align*}
  \C_{A_{\K}}(\ba_{\K})
  &\equiv
  \{\uu_{\K}: A_j\uu_j=\ba_j\ \text{for all}\ j\in\K\}.
  \\
  A_{\K}\uu_{\K}
  &\equiv\lrb{A_j\uu_j}_{j\in\K}.
\end{align*}
For $\G\subset\U^n\times\V^n$ and $\uu\in\U^n$, let $\G_{\U}$ and
$\G_{\V|\U}(\uu)$ be defined as
\begin{align*}
  \G_{\U}&\equiv\{\uu: (\uu,\vv)\in\G\ \text{for some}\ \vv\in\V^n\}
  \\
  \G_{\V|\U}(\uu)&\equiv\{\vv: (\uu,\vv)\in\G\}.
\end{align*}
Furthermore, to shorten the description of the following lemma, we use
the following abbreviation
\begin{equation}
  \lrbar{\G_{\J|\J^c}}
  \equiv
  \begin{cases}
    |\G|, &\text{if}\ \J=\K
    \\
    \displaystyle\max_{\uu_{\J^c}\in\G_{\U_{\J^c}}}
    \lrbar{\G_{\U_{\J}|\U_{\J^c}}\lrsb{\uu_{\J^c}}}
    &\text{otherwise}.
  \end{cases}
  \label{eq:maxT}
\end{equation}
for $\G\subset[\U_{\K}]^n$ and $\J\subset\K$. It should be noted that
the expression $\lrbar{\G_{\J|\J^c}}$ does not represent the cardinality
of the set $\G_{\J|\J^c}$.

\begin{lem}[{\cite[Lemma 7]{HASH-BC}}]
\label{lem:multi-CRP}
For each $j\in\K$, let $\A_j$ be a set of functions
$A_j:\U_j^n\to\im\A_j$ and $p_{\sfA_j}$ be the probability distribution
on $\A_j$, where $(\A_j,p_{\sfA_j})$ satisfies (\ref{eq:hash}). We
assume that random variables $\sfA_{\K}\equiv\{\sfA_j\}_{j\in\K}$ are
mutually independent. For each $\J\subset\K$, let $\aalpha_{\sfA_{\J}}$
and $\bbeta_{\sfA_{\J}}$ be defined as
\begin{align*}
  \alpha_{\sfA_{\J}}
  &\equiv
  \prod_{j\in\J}\alpha_{\sfA_j}
  \\
  \beta_{\sfA_{\J}}
  &\equiv
  \prod_{j\in\J}\lrB{1+\beta_{\sfA_j}}-1.
\end{align*}
Then
\begin{align*}
  p_{\sfA_{\K}}\lrsb{\lrb{
      A_{\K}: \lrB{\G\setminus\{\uu_{\K}\}}\cap\C_{A_{\K}}(A_{\K}\uu_{\K}
      )\neq \emptyset
  }}
  \leq
  \sum_{\substack{
      \J\subset\K:\\
      \J\neq\emptyset
  }}
  \frac{
    \lrbar{\G_{\J|\J^c}}
    \alpha_{\sfA_{\J}}\lrB{\beta_{\sfA_{\J^c}}+1}
  }
  {
    \prod_{j\in\J}\lrbar{\im\A_j}
  }
  +\beta_{\sfA_{\K}}
\end{align*}
for all $\G\subset\lrB{\U_{\K}}^n$ and $\uu_{\K}\in\lrB{\U_{\K}}^n$.
Furthermore, if $(\aalpha_{\sfA_j},\bbeta_{\sfA_j})$ satisfies
(\ref{eq:alpha}) and (\ref{eq:beta}) for all $j\in\K$, then
\begin{align}
  \limn \alpha_{\sfA_{\J}}(n)=1
  \label{eq:multi-alpha}
  \\
  \limn \beta_{\sfA_{\J}}(n)=0
  \label{eq:multi-beta}
\end{align}
for every $\J\subset\K$.
\end{lem}

\section{Construction of Codes}
\label{sec:mac}

In this section, we construct codes for the scenarios introduced in
Section~\ref{sec:review}.

\subsection{Private Messages}
\label{sec:mac-ts}

In this section, we consider a scenario in which $k$ senders transmit
independent messages to a receiver and there is no common message to be
sent (Fig.\ref{fig:mac}).

First, we construct a code based on the coded time-sharing technique
introduced in \cite{HK81}. For a given $\mu_{Y|X_{\K}}$, $\mu_{U}$, and
$\{\mu_{X_j|U}\}_{j\in\K}$, assume that $R_{\K}$ satisfies
(\ref{eq:sumRi-ts}). Then there is $\{\e_j\}_{j\in\K}$ such that
\begin{equation}
  \sum_{j\in\J} [R_j+\e_j]
  < I(X_{\J};Y|U,X_{\J^c})-\e
  \quad\text{for all}\ \J\subset\K,
  \label{eq:RJ-ts}
\end{equation}
where $\e$ is defined as
\begin{equation}
  \e\equiv
  \eta_{\X_{\K}|\U\Y}\lrsb{
    \left.
      2\sum_{j\in\K}\e_j
    \right|
    2\sum_{j\in\K}\e_j
  },
  \label{eq:def-e-ts}
\end{equation}
where $\eta_{\X_{\K}|\U\Y}$ is defined by (\ref{eq:def-etac}). For each
$j\in\K$, let $r_j$ be defined as
\begin{equation}
  r_j\equiv H(X_j|U)-R_j-\e_j.
  \label{eq:rjRj-ts}
\end{equation}
From (\ref{eq:markov-ts}), (\ref{eq:RJ-ts}), and (\ref{eq:rjRj-ts}), we
have
\[
  r_j
  \geq I(X_j;Y|U,X_{\K\setminus\{j\}})-R_j-\e_j
  > 0.
\]
Let $(\bcA_j,\bpA{j})$ and $(\bcAp_j,\bpAp{j})$ be ensembles of
functions, and let $\A_j\in\bcA_j$ and $\A'_j\in\bcAp_j$. Let
$A_j\in\A_j$ and $A'_j\in\A'_j$ be functions
\begin{align*}
  A_j&:\X_j^n\to\im\A_j
  \\
  A'_j&:\X_j^n\to\im\A'_j,
\end{align*}
respectively. We assume that ensembles satisfy
\begin{align}
  r_j&=\frac{\log|\im\A_j|}n
  \label{eq:rj}
  \\
  R_j&=\frac{\log|\im\A'_j|}n.
  \label{eq:Rj}
\end{align}
For each $j\in\K$, let $\M_j$ be the set of messages defined as
\[
  \M_j\equiv\im\A'_j.
\]
Then $R_j$ represents the encoding rate of the $j$-th message.
We assume that the $j$-th encoder and a decoder share functions
$A_j\in\A_j$, $A'_j\in\A'_j$ and vectors $\ba_j\in\im\A_j$ and
$\uu\in\U^n$.

For each $j\in\K$, we define the $j$-th encoder as
\[
  \Encoder_j(\mm_j)
  \equiv
  \hg_{A_jA'_j}(\ba_j,\mm_j|\uu)
\]
for a message $\mm_j\in\M_j$, where $\C_{A_jA'_j}(\ba_j,\mm_j)$ is
defined as
\begin{equation}
  \C_{A_jA'_j}(\ba_j,\mm_j)
  \equiv\{\xx_j: A_j\xx_j = \ba_j\ \text{and}\ A'_j\xx_j=\mm_j\}.
  \label{eq:CAAp}
\end{equation}
and
\[
  \hg_{A_jA'_j}(\ba_j,\mm_j|\uu)
  \equiv
  \arg\min_{\xx_j'\in\C_{A_jA'_j}(\ba_j,\mm_j)}
  D(\nu_{\xx_j'|\uu}\|\mu_{X_j|U}|\nu_{\uu}).
\]
We define the decoder as
\[
  \Decoder(\yy)
  \equiv A'_{\K}\hg_{A_{\K}}(\ba_{\K}|\yy,\uu)
\]
for a channel output $\yy\in\Y^n$, where
\[
  \hg_{A_{\K}}(\ba_{\K}|\yy,\uu)
  \equiv
  \arg\min_{\substack{
      \xx_{\K}':\\
      \xx_j'\in\C_{A_j}(\ba_j)
      \\
      \text{for all}\ j\in\K
  }}
  D(\nu_{\uu\xx_{\K}'\yy}\|\mu_{UX_{\K}Y}).
\]
Figure \ref{fig:mac-code-ts} illustrates the code construction for $k=2$.
For given vectors $\ba_j$, $\uu$, and a message $\mm_j$, the function
$\hg_{A_jA'_j}$ finds a conditionally typical sequence $\xx_j$
satisfying $A_j\xx_j=\ba_j$ and $A'_j\xx_j=\mm_j$. The function $A_j$ is
analogous to the parity check matrix for the $j$-th message, and the
function $\hg_{A_{\K}}$ is a typical set decoder that guesses the
channel input $\xx_{\K}$ satisfying $A_j\xx_j=\ba_j$ for all $j\in\K$,
where vectors $\ba_{\K}$, $\uu$, and a channel output $\yy$ are given.

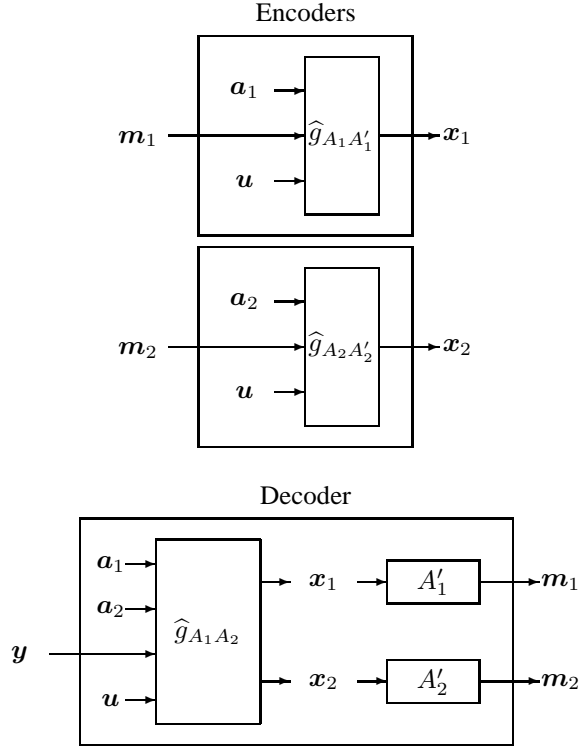
\begin{figure}[t]
  \begin{center}
    \unitlength 0.40mm
    \begin{picture}(130,99)(0,6)
      \put(65,89){\makebox(0,0){Encoders}}
      \put(45,63){\makebox(0,0){$\ba_1$}}
      \put(55,63){\vector(1,0){10}}
      \put(10,47){\makebox(0,0){$\mm_1$}}
      \put(20,48){\vector(1,0){45}}
      \put(65,22){\framebox(24,52){$\hg_{A_1A'_1}$}}
      \put(89,48){\vector(1,0){20}}
      \put(115,48){\makebox(0,0){$\xx_1$}}
      \put(45,33){\makebox(0,0){$\uu$}}
      \put(55,33){\vector(1,0){10}}
      \put(30,15){\framebox(70,66){}}
    \end{picture}
    \\
    \begin{picture}(130,61)(0,15)
      \put(45,63){\makebox(0,0){$\ba_2$}}
      \put(55,63){\vector(1,0){10}}
      \put(10,47){\makebox(0,0){$\mm_2$}}
      \put(20,48){\vector(1,0){45}}
      \put(65,22){\framebox(24,52){$\hg_{A_2A'_2}$}}
      \put(89,48){\vector(1,0){20}}
      \put(115,48){\makebox(0,0){$\xx_2$}}
      \put(45,33){\makebox(0,0){$\uu$}}
      \put(55,33){\vector(1,0){10}}
      \put(30,15){\framebox(70,66){}}
    \end{picture}
    \\
    \begin{picture}(188,99)(0,6)
      \put(94,89){\makebox(0,0){Decoder}}
      \put(30,66){\makebox(0,0){$\ba_1$}}
      \put(35,66){\vector(1,0){10}}
      \put(30,51){\makebox(0,0){$\ba_2$}}
      \put(35,51){\vector(1,0){10}}
      \put(30,21){\makebox(0,0){$\uu$}}
      \put(35,21){\vector(1,0){10}}
      \put(0,36){\makebox(0,0){$\yy$}}
      \put(10,36){\vector(1,0){35}}
      \put(45,13){\framebox(34,61){$\hg_{A_1A_2}$}}
      \put(79,60){\vector(1,0){10}}
      \put(100,60){\makebox(0,0){$\xx_1$}}
      \put(111,60){\vector(1,0){10}}
      \put(121,53){\framebox(30,14){$A'_1$}}
      \put(151,60){\vector(1,0){20}}
      \put(178,60){\makebox(0,0){$\mm_1$}}
      \put(79,27){\vector(1,0){10}}
      \put(100,27){\makebox(0,0){$\xx_2$}}
      \put(111,27){\vector(1,0){10}}
      \put(121,20){\framebox(30,14){$A'_2$}}
      \put(151,27){\vector(1,0){20}}
      \put(178,27){\makebox(0,0){$\mm_2$}}
      \put(20,6){\framebox(142,75){}}
    \end{picture}
  \end{center}
  \caption{Construction of Multiple Access Channel Code: Private Messages
    (Coded Time-sharing)}
  \label{fig:mac-code-ts}
\end{figure}

Here, let us remark on the relations (\ref{eq:RJ-ts}) and (\ref{eq:rjRj-ts}).
From these relations and (\ref{eq:markov-ts}), we have
\begin{equation}
  r_j+R_j=H(X_j|U)-\e_j
  \quad\text{for all}\  j\in\K
  \label{eq:Rjrj-SP-ts}
\end{equation}
and
\begin{align}
  \sum_{j\in\J}r_j
  &=
  \sum_{j\in\J}\lrB{H(X_j|U)-R_j-\e_j}
  \notag
  \\
  &=
  H(X_{\J}|U)-\sum_{j\in\J}\lrB{R_j+\e_j}
  \notag
  \\
  &>
  H(X_{\J}|U)-I(X_{\J};Y|U,X_{\J^c})+\e
  \notag
  \\
  &= H(X_{\J}|U,X_{\J^c},Y)+\e
  \label{eq:rj-CRP-ts}
\end{align}
for all $\J\subset\K$. Condition (\ref{eq:Rjrj-SP-ts}) is sufficient for
the saturation property, that is, for a given $\uu$ the $j$-th encoder
can find a conditionally typical sequence corresponding to the $j$-th
message $\mm_j$ when the number $2^{n[r_j+R_j]}$ of bins is smaller than
the number of typical sequences. Condition (\ref{eq:rj-CRP-ts}) is
sufficient for the collision-resistance property, that is, the decoding
error probability goes to zero if the rate vector $r_{\tK}$ of the
vector $\ba_{\K}$ is in the Slepian-Wolf region of the correlated source
coding. It should be noted that the decoder can recover messages
$\mm_{\K}$ when the channel input $\xx_{\K}$ is successfully decoded
by operating $A'_{\K}$ to $\xx_{\K}$ because the $j$-th message $\mm_j$
satisfies $A'_j\xx_j=\mm_j$.

For each $j\in\K$, let $M_j$ be a random variable corresponding to the
$j$-th message, where the probability distribution $p_{M_j}$ is uniform
on $\M_j$. Let $\Error(A_{\K},A'_{\K},\ba_{\K})$ be the decoding error
probability. We have the following theorem.
\begin{thm}
\label{thm:mac-ts}
Let $\mu_{Y|X_{\K}}$ be the conditional probability distribution
of a stationary memoryless channel and $\mu_{UX_{\K}Y}$ be defined by
(\ref{eq:markov-ts}) for given probability distributions $\mu_U$ and
$\{\mu_{X_j|U}\}_{j\in\K}$. For given
$R_{\K}\in\R(\mu_U,\{\mu_{X_j|U}\}_{j\in\K})$ and $\{\e_j\}_{j\in\K}$
satisfying (\ref{eq:rj})--(\ref{eq:def-e-ts}), assume that ensembles
$(\bcA_j,\bpA{j})$ and $(\bcAp_j,\bpAp{j})$ have a hash property for all
$j\in\K$. Then, for any $\delta>0$ and all sufficiently large $n$, there
are functions (sparse matrices) $\{A_j\}_{j\in\K}$, $\{A'_j\}_{j\in\K}$,
and vectors $\{\ba_j\}_{j\in\K}$, $\uu$ such that $A_j\in\A_j$,
$A'_j\in\A'_j$, $\ba_j\in\im\A_j$, $\uu\in\U^n$, and
\begin{equation}
  \Error(A_{\K},A'_{\K},\ba_{\K},\uu)<\delta.
  \label{eq:error-ts}
\end{equation}
\end{thm}

Next, we construct a code with $R_{\K}\in\R(\{\mu_{X_j}\}_{j\in\K})$ by
letting $U$ be a constant, that is, $|\U|=1$. Although the result is
straightforward, we describe the corollary which is used in the next
section. Condition (\ref{eq:sumRi-ts}) is replaced by (\ref{eq:sumRi}).
Condition (\ref{eq:RJ-ts}) is replaced by
\begin{equation}
  \sum_{j\in\J} [R_j+\e_j]
  < I(X_{\J};Y|X_{\J^c})-\e
  \quad\text{for all}\ \J\subset\K,
  \label{eq:RJ}
\end{equation}
where $\e$ is defined as
\begin{equation}
  \e\equiv
  \eta_{\X_{\K}|\Y}\lrsb{
    \left.
      2\sum_{j\in\K}\e_j
    \right|
    2\sum_{j\in\K}\e_j
  }.
  \label{eq:def-e}
\end{equation}
Definition (\ref{eq:rjRj-ts}) is replaced by
\begin{equation}
  r_j=H(X_j)-R_j-\e_j
  \quad\text{for all}\ j\in\K.
  \label{eq:rjRj}
\end{equation}
Functions $\hg_{A_jA'_j}$ and $\hg_{A_{\K}}$ can be replaced by
\begin{align*}
  \hg_{A_jA'_j}(\ba_j,\mm_j)
  &\equiv
  \arg\min_{\xx_j'\in\C_{A_jA'_j}(\ba_j,\mm_j)}
  D(\nu_{\xx_j'}\|\mu_{X_j})
  \\
  \hg_{A_{\K}}(\ba_{\K}|\yy)
  &\equiv
  \arg\min_{\substack{
      \xx_{\K}':\\
      \xx_j'\in\C_{A_j}(\ba_j)
      \\
      \text{for all}\ j\in\K
  }}
  D(\nu_{\xx_{\K}'\yy}\|\mu_{X_{\K}Y}),
\end{align*}
respectively, where $\C_{A_jA'_j}(\ba_j,\mm_j)$ is defined by
(\ref{eq:CAAp}).

Figure \ref{fig:mac-code} illustrates the code construction for $k=2$.
We have the following corollary.
\begin{cor}
\label{thm:mac}
Let $\mu_{Y|X_{\K}}$ be the conditional probability distribution of a
stationary memoryless channel and $\mu_{X_{\K}Y}$ be defined by
(\ref{eq:markov}) for a given $\{\mu_{X_j}\}_{j\in\K}$. For given
$R_{\K}\in\R(\{\mu_{X_j}\}_{j\in\K})$ and $\{\e_j\}_{j\in\K}$,
satisfying (\ref{eq:rj}), (\ref{eq:Rj}), and
(\ref{eq:rjRj})--(\ref{eq:def-e}), assume that ensembles
$(\bcA_j,\bpA{j})$ and $(\bcAp_j,\bpAp{j})$ have a hash property for all
$j\in\K$. Then, for any $\delta>0$ and all sufficiently large $n$, there
are functions $\{A_j\}_{j\in\K}$, $\{A'_j\}_{j\in\K}$, and vectors
$\{\ba_j\}_{j\in\K}$ such that $A_j\in\A_j$, $A'_j\in\A'_j$,
$\ba_j\in\im\A_j$, and $\Error(A_{\K},A'_{\K},\ba_{\K})<\delta$, where
$\Error(A_{\K},A'_{\K},\ba_{\K})$ denotes the error probability.
\end{cor}

\begin{figure}[t]
  \begin{center}
    \unitlength 0.40mm
    \begin{picture}(130,99)(0,6)
      \put(65,89){\makebox(0,0){Encoders}}
      \put(45,60){\makebox(0,0){$\ba_1$}}
      \put(55,60){\vector(1,0){10}}
      \put(10,36){\makebox(0,0){$\mm_1$}}
      \put(20,36){\vector(1,0){45}}
      \put(65,22){\framebox(24,52){$\hg_{A_1A'_1}$}}
      \put(89,48){\vector(1,0){20}}
      \put(115,48){\makebox(0,0){$\xx_1$}}
      \put(30,15){\framebox(70,66){}}
    \end{picture}
    \\
    \begin{picture}(130,61)(0,15)
      \put(45,60){\makebox(0,0){$\ba_2$}}
      \put(55,60){\vector(1,0){10}}
      \put(10,36){\makebox(0,0){$\mm_2$}}
      \put(20,36){\vector(1,0){45}}
      \put(65,22){\framebox(24,52){$\hg_{A_2A'_2}$}}
      \put(89,48){\vector(1,0){20}}
      \put(115,48){\makebox(0,0){$\xx_2$}}
      \put(30,15){\framebox(70,66){}}
    \end{picture}
    \\
    \begin{picture}(188,99)(0,6)
      \put(94,89){\makebox(0,0){Decoder}}
      \put(30,60){\makebox(0,0){$\ba_1$}}
      \put(35,60){\vector(1,0){10}}
      \put(30,43.5){\makebox(0,0){$\ba_2$}}
      \put(35,43.5){\vector(1,0){10}}
      \put(0,27){\makebox(0,0){$\yy$}}
      \put(10,27){\vector(1,0){35}}
      \put(45,13){\framebox(34,61){$\hg_{A_1A_2}$}}
      \put(79,60){\vector(1,0){10}}
      \put(100,60){\makebox(0,0){$\xx_1$}}
      \put(111,60){\vector(1,0){10}}
      \put(121,53){\framebox(30,14){$A'_1$}}
      \put(151,60){\vector(1,0){20}}
      \put(178,60){\makebox(0,0){$\mm_1$}}
      \put(79,27){\vector(1,0){10}}
      \put(100,27){\makebox(0,0){$\xx_2$}}
      \put(111,27){\vector(1,0){10}}
      \put(121,20){\framebox(30,14){$A'_2$}}
      \put(151,27){\vector(1,0){20}}
      \put(178,27){\makebox(0,0){$\mm_2$}}
      \put(20,6){\framebox(142,75){}}
    \end{picture}
  \end{center}
  \caption{Construction of Multiple Access Channel Code: Private Messages}
  \label{fig:mac-code}
\end{figure}
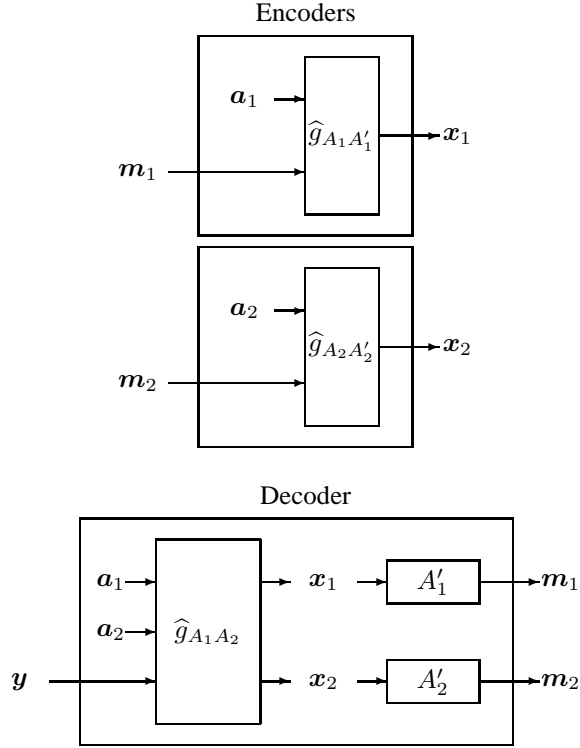

\subsection{Multiple Common Messages}
\label{sec:han}

In the following, we consider the scenario (Fig.\ref{fig:han}) where
there are $\tk$ messages and $k$ senders transmit messages common to
some users.

In the following, we assume that for given $\mu_{Y|X_{\K}}$,
$\{\mu_{\tX_i}\}_{i\in\tK}$, and $\{f_j\}_{j\in\K}$, the rate vector
$R_{\tK}$ satisfies
$R_{\tK}\in\R_{\mathrm H}(\{\mu_{\tX_i}\}_{i\in\tK},\{f_j\}_{j\in\K})$.
For a given $k$-input multiple access channel $\mu_{Y|X_{\K}}$, let us
consider a $\tk$-input multiple access channel $\mu_{Y|\tX_{\tK}}$
defined as
\[
  \mu_{Y|\tX_{\tK}}(y|\tx_{\tK})
  \equiv
  \sum_{x_{\K}}
  \mu_{Y|X_{\K}}(y|x_{\K})
  \prod_{j\in\K}\chi(f_j(\tx_{\tK_j})=x_j).
\]
Then the scenario of multiple common messages for the channel
$\mu_{Y|\X_{\K}}$ can be reduced to the scenario of private messages
for the channel $\mu_{Y|\tX_{\tK}}$ in which the $i$-th input terminal
has access to its private message $M_i$ and there is no common message.
Then, by applying Corollary \ref{thm:mac} to the channel
$\mu_{Y|\tX_{\tK}}$, we have the fact that there is a code
$(\Encoder_{\tK},\Decoder)$ for this channel at $R_{\tK}$ satisfying
(\ref{eq:sumRi-han}). Figure \ref{fig:han-reduction} illustrates the
construction of the code for the channel $\mu_{Y|\tX_{\tK}}$. A code
$(\widetilde{\Encoder}_{\K},\widetilde{\Decoder})$ for the channel
$\mu_{Y|X_{\K}}$ is given as
\begin{align*}
  \widetilde{\Encoder}_j(\mm_{\tK_j})
  &\equiv \ff_j\lrsb{\{\Encoder_i(\mm_i)\}_{i\in\tK_j}}
  \\
  \widetilde{\Decoder}(\yy)
  &\equiv
  \Decoder(\yy)
\end{align*}
for a multiple message $\mm_{\tK}$, where
\[
  \ff_j(\txx_{\tK_j})
  \equiv(f_j(\tx_{\tK_j,1}),\ldots,f_j(\tx_{\tK_j,n}))
\]
for each  $j\in\K$ and  $\txx_{\tK_j}\equiv\{\txx_i\}_{i\in\tK_j}$.
Figure \ref{fig:han-code-j} illustrates the construction of the $j$-th
encoder, where we define $\tk_j\equiv|\tK_j|$.

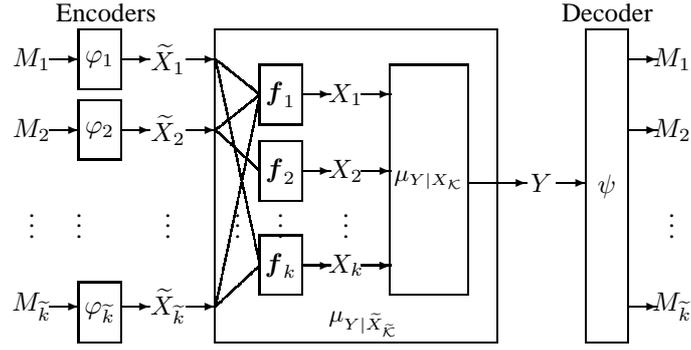
\begin{figure}
  \begin{center}
    \unitlength 0.39mm
    \begin{picture}(225,120)(-62,-6)
      \put(-32,106){\makebox(0,0){Encoders}}
      \put(-57,90){\makebox(0,0){$M_1$}}
      \put(-57,66){\makebox(0,0){$M_2$}}
      \put(-57,36){\makebox(0,0){$\vdots$}}
      \put(-57,6){\makebox(0,0){$M_{\tk}$}}
      \put(-51,90){\vector(1,0){10}}
      \put(-51,66){\vector(1,0){10}}
      \put(-51,6){\vector(1,0){10}}
      \put(-41,80){\framebox(14,20){$\Encoder_1$}}
      \put(-41,56){\framebox(14,20){$\Encoder_2$}}
      \put(-41,36){\makebox(0,0){$\vdots$}}
      \put(-41,-6){\framebox(14,20){$\Encoder_{\tk}$}}
      \put(-27,90){\vector(1,0){10}}
      \put(-27,66){\vector(1,0){10}}
      \put(-27,6){\vector(1,0){10}}
      \put(-11,90){\makebox(0,0){$\tX_1$}}
      \put(-11,66){\makebox(0,0){$\tX_2$}}
      \put(-11,36){\makebox(0,0){$\vdots$}}
      \put(-11,6){\makebox(0,0){$\tX_{\tk}$}}
      \put(-5,90){\vector(1,0){10}}
      \put(-5,66){\vector(1,0){10}}
      \put(-5,6){\vector(1,0){10}}
      \qbezier(5,90)(12.5,84)(20,78)
      \qbezier(5,90)(12.5,55)(20,20)
      \qbezier(5,66)(12.5,72)(20,78)
      \qbezier(5,66)(12.5,59)(20,52)
      \put(12.5,36){\makebox(0,0){$\vdots$}}
      \qbezier(5,6)(12.5,13)(20,20)
      \qbezier(5,6)(12.5,42)(20,78)
      \put(20,68){\framebox(14,20){$\ff_1$}}
      \put(20,42){\framebox(14,20){$\ff_2$}}
      \put(27,36){\makebox(0,0){$\vdots$}}
      \put(20,10){\framebox(14,20){$\ff_k$}}
      \put(54,20){\vector(1,0){10}}
      \put(34,78){\vector(1,0){10}}
      \put(49,78){\makebox(0,0){$X_1$}}
      \put(54,78){\vector(1,0){10}}
      \put(34,52){\vector(1,0){10}}
      \put(49,52){\makebox(0,0){$X_2$}}
      \put(54,52){\vector(1,0){10}}
      \put(49,36){\makebox(0,0){$\vdots$}}
      \put(34,20){\vector(1,0){10}}
      \put(49,20){\makebox(0,0){$X_k$}}
      \put(54,20){\vector(1,0){10}}
      \put(64,10){\framebox(26,78){\small $\mu_{Y|X_{\K}}$}}
      \put(90,48){\vector(1,0){20}}
      \put(5,-6){\framebox(95,106){}}
      \put(55,0){\makebox(0,0){\small $\mu_{Y|\tX_{\tK}}$}}
      \put(115,48){\makebox(0,0){$Y$}}
      \put(120,48){\vector(1,0){10}}
      \put(137,106){\makebox(0,0){Decoder}}
      \put(130,-6){\framebox(14,106){$\Decoder$}}
      \put(144,90){\vector(1,0){10}}
      \put(159,90){\makebox(0,0){$M_1$}}
      \put(144,66){\vector(1,0){10}}
      \put(159,66){\makebox(0,0){$M_2$}}
      \put(159,36){\makebox(0,0){$\vdots$}}
      \put(144,6){\vector(1,0){10}}
      \put(159,6){\makebox(0,0){$M_{\tk}$}}
    \end{picture}
  \end{center}
  \caption{Reduction of Multiple Common Messages to Private Messages}
  \label{fig:han-reduction}
\end{figure}

\begin{figure}
  \begin{center}
    \unitlength 0.6mm
    \begin{picture}(126,105)(0,0)
      \put(65,97){\makebox(0,0){Encoder}}
      \put(5,76){\makebox(0,0){$\mm_{i_1}$}}
      \put(5,52){\makebox(0,0){$\mm_{i_2}$}}
      \put(5,36){\makebox(0,0){$\vdots$}}
      \put(5,16){\makebox(0,0){$\mm_{i_{\tk_j}}$}}
      \put(14,76){\vector(1,0){20}}
      \put(14,52){\vector(1,0){20}}
      \put(14,16){\vector(1,0){20}}
      \put(34,66){\framebox(14,20){$\Encoder_{i_1}$}}
      \put(34,42){\framebox(14,20){$\Encoder_{i_2}$}}
      \put(44,36){\makebox(0,0){$\vdots$}}
      \put(34,6){\framebox(14,20){$\Encoder_{i_{\tk_j}}$}}
      \put(48,76){\vector(1,0){10}}
      \put(48,52){\vector(1,0){10}}
      \put(48,16){\vector(1,0){10}}
      \put(65,76){\makebox(0,0){$\txx_{i_1}$}}
      \put(65,52){\makebox(0,0){$\txx_{i_2}$}}
      \put(65,36){\makebox(0,0){$\vdots$}}
      \put(65,16){\makebox(0,0){$\txx_{i_{\tk_j}}$}}
      \put(72,76){\vector(1,0){10}}
      \put(72,52){\vector(1,0){10}}
      \put(72,16){\vector(1,0){10}}
      \put(82,6){\framebox(14,80){$\ff_j$}}
      \put(96,46){\vector(1,0){20}}
      \put(121,46){\makebox(0,0){$\xx_j$}}
      \put(24,0){\framebox(82,92){}}
    \end{picture}
  \end{center}
  \caption{Construction of $j$-th Encoder}
  \label{fig:han-code-j}
\end{figure}
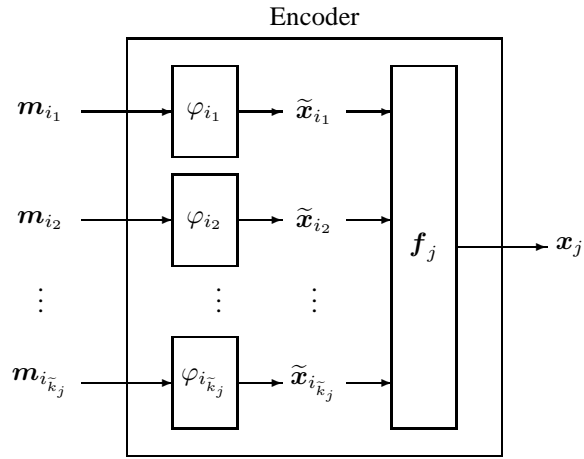

\subsection{Two-user Multiple Access Channel Coding: Private and Common
  Messages}
\label{sec:sw}

In this section we consider a scenario (Fig.\ref{fig:sw}) in which one
of two senders has access to messages $M_0$ and $M_1$ and another sender
has access to messages $M_0$ and $M_2$. We construct a code based on a
method that is analogous to a superposition coding.

For given $\mu_{Y|X_1X_2}$, $\mu_{X_1|X_0}$, $\mu_{X_2|X_0}$, and $\mu_{X_0}$,
assume that $(R_0,R_1,R_2)$ satisfies (\ref{eq:R0-sw})--(\ref{eq:R0R1R2-sw})
and the following three conditions
\begin{gather}
  R_0<I(X_0;X_1,X_2,Y)
  \label{eq:R0-sw-hash}
  \\
  R_0+R_1 < I(X_0,X_1;X_2,Y)
  \label{eq:R0R1-sw-hash}
  \\
  R_0+R_2 < I(X_0,X_2;X_1,Y),
  \label{eq:R0R2-sw-hash}
\end{gather}
which will be eliminated by using the rate-splitting technique
introduced later. Then there is $\e$ such that $\e>0$ and
\begin{gather}
  R_0+\e_0
  < I(X_0;X_1,X_2,Y)-\e
  \label{eq:R0e-sw}
  \\
  R_1+\e_1
  < I(X_1;Y|X_0,X_2)-\e
  \label{eq:R1e-sw}
  \\
  R_2+\e_1
  < I(X_2;Y|X_0,X_1)-\e
  \label{eq:R2e-sw}
  \\
  R_0+R_1+\e_0+\e_1
  < I(X_0,X_1;X_2,Y)-\e
  \label{eq:R0R1e-sw}
  \\
  R_0+R_2+\e_0+\e_2
  < I(X_0,X_2;X_1,Y)-\e
  \label{eq:R0R2e-sw}
  \\
  R_1+R_2+\e_1+\e_2
  < I(X_1,X_2;Y|X_0)-\e
  \label{eq:R1R2e-sw}
  \\
  R_0+R_1+R_2+\e_0+\e_1+\e_2
  < I(X_1,X_2;Y)-\e,
  \label{eq:R0R1R2e-sw}
\end{gather}
where $\e$ is defined by
\begin{equation}
  \e\equiv
  \eta_{\X_{\tK}|\Y}\lrsb{
    \left.
      2\sum_{j\in\tK}\e_j
    \right|
    2\sum_{j\in\tK}\e_j
  }.
  \label{eq:def-e-sw}
\end{equation}
For each $j\in\tK\equiv\{0,1,2\}$, let $r_j$ be defined as
\begin{align}
  r_0&\equiv H(X_0)-R_0-\e_0
  \label{eq:r0-sw}
  \\
  r_1&\equiv H(X_1|X_0)-R_1-\e_1
  \label{eq:r1-sw}
  \\
  r_2&\equiv H(X_2|X_0)-R_2-\e_2.
  \label{eq:r2-sw}
\end{align}
From (\ref{eq:markov-sw}), (\ref{eq:rjRj-ts}), and
(\ref{eq:R0e-sw})--(\ref{eq:R2e-sw}), we have
\begin{align*}
  r_0
  \geq I(X_0;X_1,X_2,Y)-R_0-\e_0
  >0
  \\
  r_1
  \geq I(X_1;Y|X_0,X_2)-R_1-\e_1
  >0
  \\
  r_2
  \geq I(X_2;Y|X_0,X_1)-R_2-\e_2
  >0.
\end{align*}
For $i\in\tK$, let $(\bcA_i,\bpA{i})$ and $(\bcAp_i,\bpAp{i})$ be
ensembles of functions, and let $\A_i\in\bcA_i$ and $\A'_i\in\bcAp_i$.
Let $A_i\in\A_i$ and $A'_i\in\A'_i$ be functions
\begin{align*}
  A_i&:\X_i^n\to\im\A_i
  \\
  A'_i&:\X_i^n\to\im\A'_i,
\end{align*}
respectively. We assume that ensembles satisfy
\begin{align}
  r_i&=\frac{\log|\im\A_i|}n
  \label{eq:rj-sw}
  \\
  R_i&=\frac{\log|\im\A'_i|}n.
  \label{eq:Rj-sw}
\end{align}
Let $\M_i$ be the set of messages defined as
\[
  \M_i\equiv\im\A'_i.
\]
Then $(R_0,R_1,R_2)$ represents the encoding rate of this code. We
assume that, for each $j\in\K\equiv\{1,2\}$, the $j$-th encoder and a
decoder share functions $A_0\in\A_0$, $A'_0\in\A'_0$, $A_j\in\A_j$,
$A'_j\in\A'_j$, and vectors $\ba_0\in\im\A_0$ and $\ba_j\in\im\A_j$.

Let $(\mm_0,\mm_1,\mm_2)\in\M_{\tK}$, be a multiple message.
For each $j\in\K$, we define the $j$-th encoder as
\[
  \Encoder_j(\mm_0,\mm_j)
  \equiv
  g_{A_jA'_j}(\ba_j,\mm_j|g_{A_0A'_0}(\ba_0,\mm_0)),
\]
where
\begin{align*}
  g_{A_0A'_0}(\ba_0,\mm_0)
  &\equiv
  \arg\min_{\xx_0'\in\C_{A_0A'_0}(\ba_0,\mm_0)}
  D(\nu_{\xx_0'}\|\mu_{X_0})
  \\
  g_{A_jA'_j}(\ba_j,\mm_j|\xx_0)
  &\equiv
  \arg\min_{\xx_j'\in\C_{A_jA'_j}(\ba_j,\mm_j)}
  D(\nu_{\xx_j'|\xx_0}\|\mu_{X_j|X_0}|\nu_{\xx_0}),
\end{align*}
where $\C_{A_jA'_j}(\ba_j,\mm_j)$ is defined by (\ref{eq:CAAp}). We
define the decoder as
\[
  \Decoder(\yy)
  \equiv(A'_0,A'_1,A'_2)g_{A_0A_1A_2}(\ba_0,\ba_1,\ba_2|\yy)
\]
for a channel output $\yy\in\Y^n$, where
\begin{align*}
  g_{A_0A_1A_2}(\ba_0,\ba_1,\ba_2|\yy)
  &\equiv
  \arg
  \min_{\substack{
      (\xx_0',\xx_1',\xx_2'):\\
      \xx_0'\in\C_{A_0}(\ba_0)\\
      \xx_1'\in\C_{A_1}(\ba_1)\\
      \xx_2'\in\C_{A_2}(\ba_2)
  }}
  D(\nu_{\xx_0'\xx_1'\xx_2'\yy}\|\mu_{X_0X_1X_2Y})
  \\
  (A'_0,A'_1,A'_2)(\xx_0,\xx_1,\xx_2)
  &\equiv
  (A'_0\xx_0,A'_1\xx_1,A'_2\xx_2).
\end{align*}
Figure \ref{fig:mac-code-sw} illustrates the code construction.
It should be noted that the construction is analogous to the
superposition coding introduced in \cite{C72}, where the function
$\hg_{A_0A_0'}$ finds a cloud center $\xx_0$ and the function
$\hg_{A_jA'_j}$ finds a satellite $\xx_j$ of the cloud center $\xx_0$.

\begin{figure}[t]
  \begin{center}
    \unitlength 0.40mm
    \begin{picture}(176,99)(0,6)
      \put(88,89){\makebox(0,0){Encoders}}
      \put(95,68){\makebox(0,0){$\ba_1$}}
      \put(105,68){\vector(1,0){10}}
      \put(6,47){\makebox(0,0){$\mm_1$}}
      \put(16,48){\vector(1,0){99}}
      \put(115,22){\framebox(24,52){$g_{A_1A'_1}$}}
      \put(139,48){\vector(1,0){20}}
      \put(165,48){\makebox(0,0){$\xx_1$}}
      
      \put(36,37){\makebox(0,0){$\ba_0$}}
      \put(41,37){\vector(1,0){10}}
      \put(6,18){\makebox(0,0){$\mm_0$}}
      \put(16,19){\vector(1,0){35}}
      \put(51,12){\framebox(24,32){$g_{A_0A'_0}$}}
      \put(75,28){\vector(1,0){10}}
      \put(95,28){\makebox(0,0){$\xx_0$}}
      \put(105,28){\vector(1,0){10}}
      \put(26,6){\framebox(124,75){}}
    \end{picture}
    \\
    \begin{picture}(176,81)(0,6)
      \put(95,68){\makebox(0,0){$\ba_2$}}
      \put(105,68){\vector(1,0){10}}
      \put(6,47){\makebox(0,0){$\mm_2$}}
      \put(16,48){\vector(1,0){99}}
      \put(115,22){\framebox(24,52){$g_{A_2A'_2}$}}
      \put(139,48){\vector(1,0){20}}
      \put(165,48){\makebox(0,0){$\xx_2$}}
      \put(36,37){\makebox(0,0){$\ba_0$}}
      \put(41,37){\vector(1,0){10}}
      \put(6,18){\makebox(0,0){$\mm_0$}}
      \put(16,19){\vector(1,0){35}}
      \put(51,12){\framebox(24,32){$g_{A_0A'_0}$}}
      \put(75,28){\vector(1,0){10}}
      \put(95,28){\makebox(0,0){$\xx_0$}}
      \put(105,28){\vector(1,0){10}}
      \put(26,6){\framebox(124,75){}}
    \end{picture}
    \\
    \begin{picture}(188,99)(0,6)
      \put(94,89){\makebox(0,0){Decoder}}
      \put(30,66){\makebox(0,0){$\ba_0$}}
      \put(35,66){\vector(1,0){10}}
      \put(30,51){\makebox(0,0){$\ba_1$}}
      \put(35,51){\vector(1,0){10}}
      \put(30,36){\makebox(0,0){$\ba_2$}}
      \put(35,36){\vector(1,0){10}}
      \put(0,21){\makebox(0,0){$\yy$}}
      \put(10,21){\vector(1,0){35}}
      \put(45,13){\framebox(34,61){$g_{A_0A_1A_2}$}}
      \put(79,64){\vector(1,0){10}}
      \put(100,64){\makebox(0,0){$\xx_0$}}
      \put(111,64){\vector(1,0){10}}
      \put(121,57){\framebox(30,14){$A'_0$}}
      \put(151,64){\vector(1,0){20}}
      \put(178,64){\makebox(0,0){$\mm_0$}}
      \put(79,43.5){\vector(1,0){10}}
      \put(100,43.5){\makebox(0,0){$\xx_1$}}
      \put(111,43.5){\vector(1,0){10}}
      \put(121,36.5){\framebox(30,14){$A'_1$}}
      \put(151,43.5){\vector(1,0){20}}
      \put(178,43.5){\makebox(0,0){$\mm_1$}}
      \put(79,23){\vector(1,0){10}}
      \put(100,23){\makebox(0,0){$\xx_2$}}
      \put(111,23){\vector(1,0){10}}
      \put(121,16){\framebox(30,14){$A'_2$}}
      \put(151,23){\vector(1,0){20}}
      \put(178,23){\makebox(0,0){$\mm_2$}}
      \put(20,6){\framebox(142,75){}}
    \end{picture}
  \end{center}
  \caption{Construction of Multiple Access Channel Code: Private and
    Common Messages}
  \label{fig:mac-code-sw}
\end{figure}
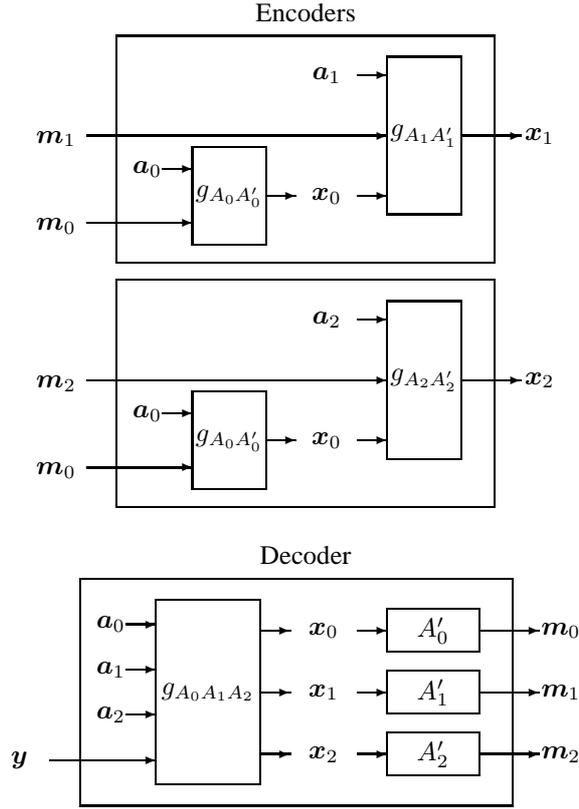

Here, we remark on relations (\ref{eq:rjRj-ts}) and
(\ref{eq:R0e-sw})--(\ref{eq:R0R1R2e-sw}). From these conditions and
(\ref{eq:markov-sw}), we have
\begin{align}
  r_0+R_0&=H(X_0)-\e_0
  \label{eq:R0r0-sw-SP}
  \\
  r_1+R_1&=H(X_1|X_0)-\e_1
  \label{eq:R1r1-sw-SP}
  \\
  r_2+R_2&=H(X_2|X_0)-\e_2
  \label{eq:R2r2-sw-SP}
\end{align}
and
\begin{align}
  r_0
  &>
  H(X_0|X_1,X_2,Y)+\e
  \label{eq:r0-sw-CRP}
  \\
  r_1
  &>
  H(X_1|X_0,X_2,Y)+\e
  \label{eq:r1-sw-CRP}
  \\
  r_2
  &>
  H(X_2|X_0,X_1,Y)+\e
  \label{eq:rj-sw-CRP}
  \\
  r_0+r_1
  &>
  H(X_0,X_1|X_2,Y)+\e
  \label{eq:r0r1-sw-CRP}
  \\
  r_0+r_2
  &>
  H(X_0,X_2|X_1,Y)+\e
  \label{eq:r0r2-sw-CRP}
  \\
  r_1+r_2
  &>
  H(X_1,X_2|X_0,Y)+\e
  \label{eq:r1r2-sw-CRP}
  \\
  r_0+r_1+r_2
  &>
  H(X_0,X_1,X_2|Y)+\e.
  \label{eq:r0r1r2-sw-CRP}
\end{align}
Conditions (\ref{eq:R0r0-sw-SP})--(\ref{eq:R2r2-sw-SP}) are sufficient
for the saturation property, that is, $\hg_{A_0A'_0}$ can find a typical
sequence corresponding to the message $\mm_0$, when the number
$2^{n[r_0+R_0]}$ of bins is smaller than the number of typical
sequences. Similarly, $\hg_{A_iA'_i}$ can find a conditionally typical
sequence for a given $\xx_0$ corresponding to the $i$-th message $\mm_i$
when the number $2^{n[r_i+R_i]}$ of bins is smaller than the number
of conditionally typical sequences. Conditions
(\ref{eq:r0-sw-CRP})--(\ref{eq:r0r1r2-sw-CRP}) are sufficient for the
collision-resistance property, that is, the decoding error probability
goes to zero if the rate $r_{\K}$ of the vector $\ba_{\K}$ is
in the Slepian-Wolf region of the correlated source coding.
It should be noted that when the channel input $(\xx_0,\xx_1,\xx_2)$ is
successfully decoded the decoder can recover the $i$-th message $\mm_i$
by operating $A'_i$ to $\xx_i$ because $\mm_i$ satisfies
$A'_i\xx_i=\mm_i$.

For each $i\in\tK$, let $M_i$ be a random variable corresponding to the
$i$-th message, where the probability distribution $p_{M_i}$ is uniform
on $\M_i$. Let $\Error(A_{\tK},A'_{\tK},\ba_{\tK})$ be the error
probability of this code. We have the following theorem.
\begin{thm}
\label{thm:sw}
Let $\mu_{Y|X_1X_2}$ be the conditional probability distribution of a
stationary memoryless channel and $\mu_{X_0X_1X_2Y}$ be defined by
(\ref{eq:markov-sw}) for given probability distributions $\mu_{X_0}$,
$\mu_{X_1|X_0}$, and $\mu_{X_2|X_0}$. For given $(r_0,r_1,r_2)$,
$(R_0,R_1,R_2)$, and $(\e_0,\e_1,\e_2)$ satisfying (\ref{eq:rjRj-ts})
and (\ref{eq:R0e-sw})--(\ref{eq:def-e-sw}), assume that ensembles
$(\bcA_j,\bpA{j})$ and $(\bcAp_j,\bpAp{j})$ have a hash property for all
$j\in\tK$. Then, for any $\delta>0$ and all sufficiently large $n$,
there are functions $\{A_j\}_{j\in\tK}$, $\{A'_j\}_{j\in\tK}$ and
vectors $\{\ba_j\}_{j\in\tK}$ such that $A_j\in\A_j$, $A'_j\in\A'_j$,
$\ba_j\in\im\A_j$, and
\begin{equation}
  \Error(A_{\tK},A'_{\tK},\ba_{\tK})<\delta.
  \label{eq:error-sw}
\end{equation}
\end{thm}

In the following, we employ a rate splitting technique to eliminate
conditions (\ref{eq:R0-sw-hash})--(\ref{eq:R0R2-sw-hash}). Assume that
$(R_0,R_1,R_2)$ satisfies conditions
(\ref{eq:R0-sw})--(\ref{eq:R0R1R2-sw}) and
(\ref{eq:R0-sw-hash})--(\ref{eq:R0R2-sw-hash}).
From Theorem~\ref{thm:sw}, we have the fact that there is a code with
encoding rate $(R_0,R_1,R_2)$. Let $\mm_0\in\X_0^{nR_0}$ be a common
message and $\mm_1\in\X_1^{nR_1}$ and $\mm_2\in\X_2^{nR_2}$ be the
private messages of two different encoders. We divide the private
messages into two parts
\begin{align*}
  \mm_1&=\lrsb{m_1^{n[R_1-R''_1]},m_1^{nR''_1}}
  \\
  \mm_2&=\lrsb{m_1^{n[R_2-R''_2]},m_2^{nR''_2}},
\end{align*}
where $(R''_1,R''_2)$ satisfies
\begin{gather}
  0\leq R''_1\leq R_1
  \label{eq:R1pp}
  \\
  0\leq R''_2\leq R_2.
  \label{eq:R2pp}
\end{gather}
Let us interpret $\lrsb{\mm_0,m_1^{nR''_1},m_2^{nR''_2}}$ as the common
message and $m_j^{n[R_j-R''_j]}$ as the private message of the $j$-th
encoder. Then we have the fact that rate $(R'_0,R'_1,R'_2)$ satisfying
\begin{align}
  R'_0&= R_0+R''_1+R''_2
  \label{eq:R0p}
  \\
  R'_1&=R_1-R''_1
  \label{eq:R1p}
  \\
  R'_2&= R_2-R''_2
  \label{eq:R2p}
\end{align}
is achievable by using the same code obtained from Theorem~\ref{thm:sw}.

Now we prove the fact that for all
$(R'_0,R'_1,R'_2)\in\R_{\mathrm SW}(\mu_{X_0},\mu_{X_1|X_0},\mu_{X_2|X_0})$
there is a pair $(R''_1,R''_2)$ such that $(R_0,R_1,R_2)$ satisfies
conditions (\ref{eq:R0-sw})--(\ref{eq:R0R1R2-sw}),
(\ref{eq:R0-sw-hash})--(\ref{eq:R0R2-sw-hash}),
(\ref{eq:R1pp})--(\ref{eq:R2p}). From (\ref{eq:R0p})--(\ref{eq:R2p}), we
have
\begin{align}
  R_0&=R'_0-R''_1-R''_2
  \label{eq:R0ppp}
  \\
  R_1&=R'_1+R''_1
  \label{eq:R1ppp}
  \\
  R_2&=R'_2+R''_2.
  \label{eq:R2ppp}
\end{align}
By substituting these inequalities into
(\ref{eq:R0-sw})--(\ref{eq:R0R1R2-sw}),
(\ref{eq:R0-sw-hash})--(\ref{eq:R0R2-sw-hash}), (\ref{eq:R1pp}), and
(\ref{eq:R2pp}), we have 
\begin{gather*}
  R'_0-R''_1-R''_2 \geq 0
  \\
  0\leq R'_1+R''_1 < I(X_1;Y|X_0,X_2)
  \\
  0\leq R'_2+R''_2 < I(X_2;Y|X_0,X_1)
  \\
  R'_1+R''_1+R'_2+R''_2 < I(X_1,X_2;Y|X_0)
  \\
  R'_0+R'_1+R'_2 < I(X_0,X_1,X_2;Y)
  \\
  R'_0-R''_1-R''_2< I(X_0;X_1,X_2,Y)
  \\
  R'_0+R'_1-R''_2 < I(X_0,X_1;X_2,Y)
  \\
  R'_0+R'_2-R''_1 < I(X_0,X_2;X_1,Y)
  \\
  0\leq R''_1\leq R'_1+R''_1
  \\
  0\leq R''_2\leq R'_2+R''_2,
\end{gather*}
where we use the relation $I(X_0,X_1,X_2;Y)=I(X_1,X_2;Y)$ obtained from
(\ref{eq:markov-sw}) in the fifth inequality. By eliminating $R''_1$ and
$R''_2$ from these inequalities by using the Fourier-Motzkin method
(see~\cite[Appendix~D]{EK10}\cite{Z06}) and the relation
$I(X_1;X_2|X_0)=0$ obtained from (\ref{eq:markov-sw}), we have the fact
that
$(R'_0,R'_1,R'_2)\in\R_{\mathrm SW}(\mu_{X_0},\mu_{X_1|X_0},\mu_{X_2|X_0})$.
This implies that for 
$(R'_0,R'_1,R'_2)\in\R_{\mathrm SW}(\mu_{X_0},\mu_{X_1|X_0},\mu_{X_2|X_0})$,
there is $(R''_1,R''_2)$ such that $(R_0,R_1,R_2)$ defined by
(\ref{eq:R0ppp})--(\ref{eq:R2ppp}) satisfies 
$(R_0,R_1,R_2)\in\R_{\mathrm SW}(\mu_{X_0},\mu_{X_1|X_0},\mu_{X_2|X_0})$,
(\ref{eq:R0-sw-hash})--(\ref{eq:R0R2-sw-hash}), (\ref{eq:R1pp}), and
(\ref{eq:R2pp}). This means that we can construct codes with
$(R'_0,R'_1,R'_2)\in\R_{\mathrm SW}(\mu_{X_0},\mu_{X_1|X_0},\mu_{X_2|X_0})$.
Thus, conditions (\ref{eq:R0-sw-hash})--(\ref{eq:R0R2-sw-hash}) are
eliminated.

\section{Proof of Theorems}
\label{sec:proof}

In this section, we prove Theorems~\ref{thm:mac-ts} and~\ref{thm:sw}.
Before describing the proof, we remark on the outline of the proof.
The proof is similar to the conventional random coding argument, where
a codebook is randomly generated and it is proved that the average error
probability tends to zero as $n$ goes to infinity. However, there is a
definite difference from the conventional random coding argument in the
following proof because our proof is based on random partitioning and
the probability distribution of a codebook is different. This will be
explained in detail later.

In the following proof, we omit the dependence of $X,Y,U$ on $n$ when
they appear in the subscript of $\mu$. For
$\uu\equiv(u_1,\ldots,u_n)\in\U^n$ and 
$(\xx_{\K},\yy)\equiv((\{x_{1,j}\}_{j\in\K},y_1),\ldots,(\{x_{n,j}\}_{j\in\K},y_n))\in[\X_{\K}]^n\times\Y^n$,
$\mu_U(\uu)$ and $\mu_{Y|X_{\K}}(\yy|\xx_{\K})$ are defined as
\begin{align*}
  \mu_U(\uu)
  &\equiv\prod_{i=1}^n \mu_U(u_i)
  \\
  \mu_{Y|X_{\K}}(\yy|\xx_{\K})
  &\equiv
  \prod_{i=1}^n \mu_{Y|X_{\K}}(y_i|\{x_{i,j}\}_{j\in\K}).
\end{align*}

\subsection{Proof of Theorem \ref{thm:mac-ts}}

In the following, we assume that ensembles $(\bcA_j,\bpA{j})$ and
$(\bcAp_j,\bpAp{j})$ have a hash property for all $j\in\K$. Then, from
Lemma \ref{lem:hash-AB}, ensemble $(\bhcA_j,\bphA{j})$ defined by
\[
  \hA_j\xx_j
  \equiv(A_j\xx_j,A_j'\xx_j)
\]
has a $(\aalphahA{j},\bbetahA{j})$-hash property, where
\begin{align*}
  p_{\hA_j}(\hA_j)&\equiv p_{A_j}(A_j)p_{A'_j}(A'_j)
  \\
  \alphahA{j}&\equiv \alphaA{j}\alphaAp{j}
  \\
  \betahA{j}&\equiv \betaA{j}+\betaAp{j}.
\end{align*}

Since
\begin{align}
  \limn\beta_{\sfA_{\K}}(n)
  &=\limn\lrB{\prod_{j\in\K}[\beta_{\sfA_j}+1]-1}
  \notag
  \\
  &=0,
\end{align}
there is a sequence $\kkappa\equiv\{\kappa(n)\}_{n=1}^{\infty}$ such
that
\begin{gather}
  \limn\kappa(n)=\infty
  \label{eq:k1}
  \\
  \limn [\kappa(n)]^{k}\beta_{\sfA_{\K}}(n)=0
  \label{eq:k2}
  \\
  \limn\frac{\log\kappa(n)}n=0,
  \label{eq:k3}
\end{gather}
where $k$ is the number of encoders. For example, we obtain such a
$\kkappa$ by letting
\begin{equation*}
  \kappa(n)
  \equiv
  \begin{cases}
    n^{\xi/k}
    &\text{if}\ \exists\xi>0\ \text{s.t.}
    \ \beta_{\sfA_{\K}}(n)=o\lrsb{n^{-\xi/k}}
    \\
    \lrB{\beta_{\sfA_{\K}}(n)}^{-1/[k+1]}
    &\text{otherwise}
  \end{cases}
\end{equation*}
for every $n$. If $\beta_{\sfA_{\K}}(n)$ is not $o\lrsb{n^{-\xi/k}}$,
there is a $\kappa'$ such that $\kappa'>0$,
$\beta_{\sfA_{\K}}(n)n^{\xi/k}>\kappa'$ and
\begin{align}
  \frac{\log\kappa(n)}n
  &=
  \frac{\log\frac 1{\beta_{\sfA_{\K}}(n)}}{[k+1]n}
  \notag
  \\
  &\leq \frac{\log\frac{n^{\xi/k}}{\kappa'}}{[k+1]n}
  \notag
  \\
  &=\frac{\xi\log n}{k[k+1]n}-\frac{\log\kappa'}{[k+1]n}
\end{align}
for all sufficiently large $n$. This implies that $\kkappa$ satisfies
(\ref{eq:k3}). In the following, $\kappa$ denotes $\kappa(n)$.

From (\ref{eq:k3}), we have the fact that there is a $\gamma$ such that
$\gamma>0$ and
\begin{gather}
  \eta_{\X_j|\U}(\gamma|\gamma)+\frac{\log\kappa}n\leq \e_j
  \label{eq:ekappa}
  \\
  [k+3]\gamma+\sum_{j\in\K}\iota_{\X_j|\U}(\gamma|\gamma)
  \leq\sum_{j\in\K}\e_j
  \label{eq:sum-e}
\end{gather}
for all $j\in\K$ and sufficiently large $n$.

When $\uu\in\T_{U,\gamma}$, we have
\begin{align}
  |\T_{X_j|U,\gamma}(\uu)|
  &\geq 2^{n[H(X_j|U)-\eta_{\X_j|\U}(\gamma|\gamma)]}
  \notag
  \\
  &\geq
  \kappa2^{n[H(X_j|U)-\e_j]}
  \notag
  \\
  &=
  \kappa2^{n[r_j+R_j]}
  \notag
  \\
  &=
  \kappa|\im\A_j||\im\A'_j|
  \notag
  \\
  &\geq
  \kappa|\im\hcA_j|
\end{align}
for all $j\in\K$ and sufficiently large $n$, where the first inequality
comes from Lemma \ref{lem:typical-number-bound}, the second inequality
comes from (\ref{eq:ekappa}), the first equality comes from
(\ref{eq:rj}) and (\ref{eq:Rj}), and the last inequality comes from the
fact that $\im\hcA_j\subset\im\A_j\times\im\A'_j$.

This implies that for all $j\in\K$ and sufficiently large $n$ there is
$\T_j(\uu)\subset\T_{X_j|U,\gamma}(\uu)$ such that
\begin{equation}
  \kappa
  \leq
  \frac{|\T_j(\uu)|}{|\im\hcA_j|}
  \leq 
  2\kappa.
  \label{eq:Tj}
\end{equation}
for all $\uu$. We assume that $\T_j(\uu)$ is constructed by selecting
$|\T_j(\uu)|$ elements in the ascending order regarding the value
$D(\nu_{\xx_j|\uu}\|\mu_{X_j|U}|\nu_{\uu})$.

Let $\mm_{\K}\in\M_{\K}$ be private messages. Let $\xx_{\K}$ be channel
inputs, where $\xx_j\in\X^n_j$ is defined as
\[
  \xx_j\equiv \hg_{A_jA'_j}(\ba_j,\mm_j|\uu)
  \quad\text{for each}\ j\in\K.
\]
Let $\yy\in\Y^n$ be a channel output. We define
\begin{align*}
  \cS_j
  &\equiv
  \lrb{(\mm_{\K},\yy):
    \xx_j\in\T_j(\uu)\ \text{and}\ \yy\in\Y^n
  }
  \\
  \cS_{k+1}
  &\equiv
  \lrb{
    (\mm_{\K},\yy):
    I(\xx_{\K}|\uu)
    <\gamma+\sum_{j\in\K}\lrB{\iota_{\X_j|\U}(\gamma|\gamma)+\e_j}
    \ \text{and}\ \yy\in\Y^n
  }
  \\
  \cS_{k+2}
  &\equiv
  \lrb{(\mm_{\K},\yy):
    \yy\in\T_{Y|UX_{\K},\gamma}(\uu,\xx_{\K})
  }
  \\
  \cS_{k+3}
  &\equiv
  \lrb{(\mm_{\K},\yy):
    \hg_{A_{\K}}(\ba_{\K}|\yy,\uu)=\xx_{\K}
  },
\end{align*}
where $j\in\K$ and
\[
  I(\xx_{\K}|\uu)\equiv \sum_{j\in\K}H(\xx_j|\uu)-H(\xx_{\K}|\uu).
\]
We assign equation numbers to the conditions
\begin{gather}
  \uu\in\T_{U,\gamma}
  \label{eq:MAC0}
  \\
  \xx_j\in\T_j(\uu)\subset\T_{X_j|U,\gamma}(\uu)\quad\text{for all}\ j\in\K
  \label{eq:MACj}
  \\
  I(\xx_{\K}|\uu)<\gamma+\sum_{j\in\K}\lrB{\iota_{\X_j|\U}(\gamma|\gamma)+\e_j}
  \label{eq:MACk+1}
  \\
  \yy\in\T_{Y|UX_{\K},\gamma}(\uu,\xx_{\K})
  \label{eq:MACk+2}
  \\
  \hg_{A_{\K}}(\ba_{\K}|\yy,\uu)\neq \xx_{\K},
  \label{eq:MACk+3}
\end{gather}
which are referred later. Since the $j$-th message $\mm_j$ satisfies
$A'_j\xx_j=\mm_j$, the decoder can recover message $\mm_{\K}$ when
decoding the channel input $\xx_{\K}$ is successful. This implies that
the decoding error probability is upper bounded by
\begin{align}
  &
  \Error(A_{\K},A'_{\K},\ba_{\K},\uu)
  \notag
  \\*
  &
  \leq
  \sum_{j\in\K}p_{M_{\K}Y}(\cS_j^c)
  +p_{M_{\K}Y}\lrsb{\lrB{\cap_{j=1}^k\cS_j}\cap\cS_{k+1}^c}
  +p_{M_{\K}Y}(\cS_{k+2}^c)
  +p_{M_{\K}Y}\lrsb{\lrB{\cap_{j=1}^{k+2}\cS_j}\cap\cS_{k+3}^c}.
  \label{eq:error-sum}
\end{align}
\begin{rem}
The condition (\ref{eq:MACk+1}) was unnecessary in the proof of the
conventional random coding argument because
$\xx_{\K}\in\T_{X_{\K|U},\gamma}(\uu)$ was naturally satisfied by
generating codewords independently at random for a given $\uu\in\U^n$.
On the other hand, (\ref{eq:MACk+1}) is necessary in our proof because
(\ref{eq:MAC0}) and (\ref{eq:MACj}) may not imply
$(\uu,\xx_{\K})\in\T_{UX_{\K},\gamma'}$ for an appropriate $\gamma'>0$.
This is the difference from the conventional proof. It should be noted
that (\ref{eq:MAC0})--(\ref{eq:MACk+2}) implies
$(\uu,\xx_{\K},\yy)\in\T_{UX_{\K}Y,\gamma'}$, where $\gamma'>0$ will be
specified later.
\end{rem}

In the following we evaluate the average error probability
\begin{align}
  &
  E_{\sfhA_{\K}\sfaa_{\K}U^n}\lrB{
    \Error(\sfA_{\K},\sfA'_{\K},\sfaa_{\K},U^n)
  }
  \notag
  \\*
  &\leq
  E_{\sfhA_{\K}\sfaa_{\K}}\lrB{
    \sum_{\uu\in\T_{U,\gamma}}\mu_U(\uu)
    \Error(\sfA_{\K},\sfA'_{\K},\sfaa_{\K},\uu)
  }
  +
  \mu_U\lrsb{\lrB{\T_{U,\gamma}}^c}
  \notag
  \\
  \begin{split}
    &\leq
    \sum_{j\in\K}
    E_{\sfhA_{\K}\sfaa_{\K}}\lrB{
      \sum_{\uu\in\T_{U,\gamma}}\mu_U(\uu)
      p_{M_{\K}Y}(\cS_j^c)
    }
    +
    E_{\sfhA_{\K}\sfaa_{\K}}\lrB{
      \sum_{\uu\in\T_{U,\gamma}}\mu_U(\uu)
      p_{M_{\K}Y}\lrsb{\lrB{\cap_{j=1}^k\cS_j}\cap\cS_{k+1}^c}
    }
    \\*
    &\quad
    +
    E_{\sfhA_{\K}\sfaa_{\K}}\lrB{
      \sum_{\uu\in\T_{U,\gamma}}\mu_U(\uu)
      p_{M_{\K}Y}(\cS_{k+2}^c)
    }
    +
    E_{\sfhA_{\K}\sfaa_{\K}}\lrB{
      \sum_{\uu\in\T_{U,\gamma}}\mu_U(\uu)
      p_{M_{\K}Y}\lrsb{\lrB{\cap_{j=1}^{k+2}\cS_j}\cap\cS_{k+3}^c}
    }
    \\*
    &\quad
    +
    \mu_U\lrsb{\lrB{\T_{U,\gamma}}^c},
  \end{split}
  \label{eq:error-ave}
\end{align}
over random variables $\sfhA_{\K}$, $\sfaa_{\K}$, and $U^n$. The last
term on the right hand side of (\ref{eq:error-ave}) is evaluated as
\begin{align}
  \mu_U\lrsb{\lrB{\T_{U,\gamma}}^c}
  &\leq
  2^{-n[\gamma-\lambda_{\U}]}
  \notag
  \\
  &
  \leq
  \frac {\delta}{k+4}
  \label{eq:error0}
\end{align}
for all $\delta>0$ and all sufficiently large $n$, where the first
inequality comes from Lemma~\ref{lem:typical-prob}. In the following,
let
\[
  \haa_j\equiv (\ba_j,\mm_j)
  \quad\text{for each}\  j\in\K.
\]
We assume that the distribution of $\haa_j$ is uniform on $\im\hcA_j$
for all $j\in\K$, and random variables
$\{\sfhA_j,\sfaa_j,M_j\}_{j\in\K}$ and $U^n$ are mutually independent.
In the following, we use the fact that
$\hg_{A_jA'_j}(\ba_j,\mm_j|\uu)\notin\T_j(\uu)$ implies
$\T_j(\uu)\cap\C_{\hA_j}(\haa_j)=\emptyset$, which is shown by
contradiction as follows. Let us assume that
$\xx_j\equiv\hg_{A_jA'_j}(\ba_j,\mm_j|\uu)\notin\T_j(\uu)$
and $\T_j(\uu)\cap\C_{\hA_j}(\haa_j)\neq\emptyset$.
Then there is $\xx'_j\in\T_j(\uu)\cap\C_{\hA_j}(\haa_j)$. From the
definition of $\hg_{A_jA'_j}$, we have
\begin{equation}
  D(\nu_{\xx_j|\uu}\|\mu_{X_j|U}|\nu_{\uu})
  \leq
  D(\nu_{\xx'_j|\uu}\|\mu_{X_j|U}|\nu_{\uu}).
  \label{eq:proof-mac-DxDxp}
\end{equation}
On the other hand, from the construction of $\T_j(\uu)$, we have the
fact that $\xx''_j\in\T_j(\uu)$ if $\xx'_j\in\T_j(\uu)$ and
\[
  D(\nu_{\xx''_j|\uu}\|\mu_{X_j|U}|\nu_{\uu})
  \leq
  D(\nu_{\xx'_j|\uu}\|\mu_{X_j|U}|\nu_{\uu}).
\]
From  this fact and (\ref{eq:proof-mac-DxDxp}), we have the fact that
$\xx_j\in\T_j(\uu)$, which contradicts the assumption
$\xx_j\equiv\hg_{A_jA'_j}(\ba_j,\mm_j|\uu)\notin\T_j(\uu)$.
First, we evaluate
$
E_{\sfhA_{\K}\sfaa_{\K}}\lrB{
  \sum_{\uu\in\T_{U,\gamma}}\mu_U(\uu)
  p_{M_{\K}Y}(\cS_j^c)
}.
$
From Lemma~\ref{lem:SP} and (\ref{eq:Tj}), we have
\begin{align}
  E_{\sfhA_{\K}\sfaa_{\K}}\lrB{
    \sum_{\uu\in\T_{U,\gamma}}\mu_U(\uu)
    p_{M_{\K}Y}(\cS_j^c)
  }
  &=
  \sum_{\uu\in\T_{U,\gamma}}\mu_U(\uu)
  p_{\sfhA_j\sfaa_jM_j}
  \lrsb{\lrb{
      (A_j,A'_j,\ba_j,\mm_j):
      \hg_{A_jA'_j}(\ba_j,\mm_j|\uu)\notin\T_j(\uu)
  }}
  \notag
  \\
  &\leq
  \sum_{\uu\in\T_{U,\gamma}}\mu_U(\uu)
  p_{\sfhA_j\sfhaa_j}
  \lrsb{\lrb{
      (\hA_j,\haa_j):
      \T_j(\uu)\cap\C_{\hA_j}(\haa_j)=\emptyset
  }}
  \notag
  \\
  &\leq
  \sum_{\uu\in\T_{U,\gamma}}\mu_U(\uu)
  \lrB{
    \alpha_{\sfhA_j}-1
    +\frac{|\im\hcA_j|\lrB{\beta_{\sfhA_j}+1}}
    {|\T_j(\uu)|}
  }
  \notag
  \\
  &\leq
  \sum_{\uu\in\T_{U,\gamma}}\mu_U(\uu)
  \lrB{
    \alpha_{\sfhA_j}-1
    +\frac{\beta_{\sfhA_j}+1}
    {\kappa}
  }
  \notag
  \\
  &\leq
  \frac {\delta}{k+4}
  \label{eq:errorj}
\end{align} 
for all $\delta>0$ and sufficiently large $n$, where the first inequality
comes from the fact that $\hg_{A_jA'_j}(\ba_j,\mm_j|\uu)\notin\T_j(\uu)$
implies $\T_j(\uu)\cap\C_{\hA_j}(\haa_j)=\emptyset$, and the last
inequality comes from (\ref{eq:k1}) and the fact that $\alphahA{j}\to1$
and $\betahA{j}\to0$ as $n\to\infty$.

Next, we evaluate the second term on the right hand side of
(\ref{eq:error-ave}). Assume that $\xx_{\K}$ satisfies
(\ref{eq:MAC0}), (\ref{eq:MACj}), and
\[
  I(\xx_{\K}|\uu)
  \geq
  \gamma+
  \sum_{j\in\K}\lrB{\iota_{\X_j|\U}(\gamma|\gamma)+\e_j}.
\]
Then, from Lemma~\ref{lem:diff-entropy}, we have
\[
  |H(\xx_j|\uu)-H(X_j|U)|
  <\iota_{\X_j|\U}(\gamma|\gamma)
  \quad\text{for all}\ j\in\K.
\]
We have
\begin{align}
  H(\xx_{\K}|\uu)
  &=\sum_{j\in\K}H(\xx_j|\uu)-I(\xx_{\K}|\uu)
  \notag
  \\
  &\leq
  \sum_{j\in\K}\lrB{H(X_j|U)+\iota_{\X_j|\U}(\gamma|\gamma)}
  -\lrB{\gamma+\sum_{j\in\K}\lrB{\iota_{\X_j|\U}(\gamma|\gamma)+\e_j}}
  \notag
  \\
  &=\sum_{j\in\K}[r_j+R_j]-\gamma,
\end{align}
where the last equality comes from (\ref{eq:rjRj-ts}). Since
$\xx_j\in\C_{A_jA'_j}(\ba_j,\mm_j)$ for all $j\in\K$, we have
\[
  \xx_{\K}
  \in
  \G(\uu)\cap\C_{\hA_{\K}}(\haa_{\K}),
\]
where $\G(\uu)\subset\Prod_{j\in\K}\X_j^n$ is defined as
\[
  \G(\uu)\equiv
  \lrb{
    \xx_\K:
    H(\xx_{\K}|\uu)<\sum_{j\in\K}[r_j+R_j]-\gamma
  }.
\]
This implies that
\[
  \G(\uu)\cap\C_{\hA_{\K}}(\haa_{\K})
  \neq\emptyset.
\]
Then we have
\begin{align}
  p_{\sfhA_{\K}\sfhaa_{\K}}
  \lrsb{\lrb{
      (\hA_{\K},\haa_{\K}):
      \G(\uu)\cap\C_{\hA_{\K}}(\haa_{\tK})\neq\emptyset
  }}
  &\leq
  \sum_{\xx_{\K}\in\G(\uu)}
  p_{\sfhA_{\K}\sfhaa_{\K}}
  \lrsb{\lrb{
      (\hA_{\K},\haa_{\K}):
      \hA_j\xx_j=\haa_j
      \text{for all}\ j\in\K
  }}
  \notag
  \\
  &=
  \sum_{\xx_{\K}\in\G(\uu)}
  \sum_{\hA_{\K},\haa_{\K}}
  \prod_{j\in\K}
  p_{\sfhA_j\sfhaa_j}(\hA_j,\haa_j)\chi\lrsb{\hA_j\xx_j=\haa_j}
  \notag
  \\
  &=
  \sum_{\xx_{\K}\in\G(\uu)}
  \prod_{j\in\K}
  \lrB{
    \sum_{\hA_j,\haa_j}
    p_{\sfhA_j\sfhaa_j}(\hA_j,\haa_j)\chi\lrsb{\hA_j\xx_j=\haa_j}
  }
  \notag
  \\
  &=
  \sum_{\xx_{\K}\in\G(\uu)}
  \frac 1{\prod_{j\in\K}|\im\hcA_j|}
  \notag
  \\
  &=
  \frac
  {|\G(\uu)|}
  {\prod_{j\in\K}|\im\hcA_j|}
  \notag
  \\
  &\leq
  \frac
  {2^{n\lrB{\sum_{j\in\K}[r_j+R_j]-\gamma+\lambda_{\X_{\K}}}}}
  {\prod_{j\in\K}|\im\hcA_j|}
  \notag
  \\
  &=
  2^{-n[\gamma-\lambda_{\X_{\K}}]},
\end{align}
where the second inequality comes from Lemma~\ref{lem:typical-number}.
This implies that
\begin{align}
  E_{\sfhA_{\K}\sfaa_{\K}}\lrB{
    \sum_{\uu\in\T_{U,\gamma}}\mu_U(\uu)
    p_{M_{\K}Y}([\cap_{j=1}^k\cS_j]\cap\cS_{k+1}^c)
  }
  &\leq
  \sum_{\uu\in\T_{U,\gamma}}\mu_U(\uu)
  p_{\sfhA_{\K}\sfhaa_{\K}}
  \lrsb{\lrb{
      (\hA_{\K},\haa_{\K}):
      \G(\uu)\cap\C_{\hA_{\K}}(\haa_{\tK})\neq\emptyset
  }}
  \notag
  \\
  &\leq
  \sum_{\uu\in\T_{U,\gamma}}
  \mu_U(\uu)
  2^{-n[\gamma-\lambda_{\X_{\K}}]}
  \notag
  \\
  &\leq\frac{\delta}{k+4}
  \label{eq:error[k+1]}
\end{align} 
for all $\delta>0$ and sufficiently large $n$, where the last inequality
comes from the fact that $\lambda_{\X_{\K}}\to 0$ as $n\to\infty$.

Next, we evaluate the third term on the right hand side of
(\ref{eq:error-ave}). Let
$\XX_{\K}\equiv \{\hg_{\sfA_j\sfA'_j}(\sfaa_j,M_j|\uu)\}_{j\in\K}$.
Then we have
\begin{align}
  \mu_{Y|X_{\K}}\lrsb{
    \lrB{\T_{Y|UX_{\K},\gamma}(\uu,\XX_{\K})}^c
    |\XX_{\K}
  }
  &=
  \mu_{Y|U^nX_{\K}}\lrsb{
    \lrB{\T_{Y|UX_{\K},\gamma}(\uu,\XX_{\K})}^c
    |\uu,\XX_{\K}
  }
  \notag
  \\*
  &\leq
  2^{-n[\gamma-\lambda_{\U\X_{\K}\Y}]}.
\end{align}
from Lemma~\ref{lem:typical-prob}. This implies that
\begin{align}
  E_{\sfhA_{\K}\sfaa_{\K}}\lrB{
    \sum_{\uu\in\T_{U,\gamma}}\mu_U(\uu)
    p_{M_{\K}Y}(\cS_{k+2}^c)
  }
  &
  =
  E_{\sfhA_{\K}\sfhaa_{\K}}\left[
    \sum_{\uu\in\T_{U,\gamma}}\mu_U(\uu)
    \mu_{Y|X_{\K}}\lrsb{
      \lrB{\T_{Y|UX_{\K},\gamma}(\uu,\XX_{\K})}^c
      |\XX_{\K}
    }
  \right]
  \notag
  \\
  &\leq
  2^{-n[\gamma-\lambda_{\U\X_{\K}\Y}]}
  \notag
  \\
  &
  \leq
  \frac {\delta}{k+4}
  \label{eq:error[k+2]}
\end{align}
for all $\delta>0$ and sufficiently large $n$, where the last inequality
comes from the fact that $\lambda_{\U\X_{\K}\Y}\to 0$ as $n\to\infty$.

Next, we evaluate the fourth term on the right hand side of
(\ref{eq:error-ave}). In the following, we assume that
(\ref{eq:MAC0})--(\ref{eq:MACk+2}) and
\[
  \hg_{A_{\K}}(\ba_{\K}|\yy,\uu)\neq \xx_{\K}.x
\]
Then, from (\ref{eq:markov-ts}), we have
\begin{align}
  D(\nu_{\uu\xx_{\K}\yy}\parallel\mu_{UX_{\K}Y})
  &=
  \sum_{u,x_{\K},y}\nu_{\uu\xx_{\K}\yy}(u,x_{\K},y)
  \log\frac{\nu_{\uu\xx_{\K}\yy}(u,x_{\K},y)}
  {\mu_{UX_{\K}Y}(u,x_{\K},y)}
  \notag
  \\
  &=
  \sum_{u,x_{\K},y}\nu_{\uu\xx_{\K}\yy}(u,x_{\K},y)
  \log\frac{\nu_{\yy|\uu\xx_{\K}}(y|u,x_{\K})}
  {\mu_{Y|UX_{\K}}(y|u,x_{\K})}
  +
  \sum_{j\in\K}
  \sum_{u,x_j}\nu_{\uu\xx_j}(u,x_j)
  \log\frac{\nu_{\xx_j|\uu}(x_j|u)}{\mu_{X_j|U}(x_j|u)}
  \notag
  \\*
  &\quad
  +
  \sum_{u}\nu_{\uu}(u)
  \log\frac{\nu_{\uu}(u)}{\mu_U(u)}
  +
  \sum_{x_{\K}}\nu_{\uu\xx_{\K}}(u,x_{\K})
  \log\frac{\nu_{\xx_{\K}|\uu}(x_{\K}|u)}
  {\prod_{j\in\K}\nu_{\xx_j|\uu}(x_j|u)}
  \notag
  \\
  &
  =
  D(\nu_{\yy|\uu\xx_{\K}}\parallel\mu_{Y|UX_{\K}}|\nu_{\uu\xx_{\K}})
  +\sum_{j\in\K}D(\nu_{\xx_j|\uu}\parallel\mu_{X_j|U}|\nu_{\uu})
  +D(\nu_{\uu}\|\mu_U)+I(\xx_{\K}|\uu)
  \notag
  \\
  &<[k+2]\gamma+\gamma+\sum_{j\in\K}\lrB{\iota_{\X_j|\U}(\gamma|\gamma)+\e_j}
  \notag
  \\
  &
  \leq 2\sum_{j\in\K}\e_j
  \label{eq:D}
\end{align}
where the last inequality comes from (\ref{eq:sum-e}). This implies that
\[
  (\uu,\xx_{\K},\yy)\in\T_{X_{\K}Y,\gamma'},
\]
where $\gamma'$ is defined as
\[
  \gamma'\equiv 2\sum_{j\in\K}\e_j.
\]
Since $\hg_{A_{\K}}(\ba_{\K}|\yy,\uu)\neq\xx_{\K}$, there is
$\xx_{\K}'\in\C_{A_{\K}}(\ba_{\K})$ such that $\xx_{\K}'\neq\xx_{\K}$
and $(\uu,\xx_{\K}',\yy)\in\T_{UX_{\K}Y,\gamma'}$.
This implies that
\[
  \lrB{\G(\uu,\yy)\setminus\{\xx_{\K}\}}
  \cap\C_{A_{\K}}(A_{\K}\xx_{\K})
  \neq\emptyset,
\]
where 
\[
  \G(\uu,\yy)\equiv
  \lrb{\xx_{\K}:(\uu,\xx_{\K},\yy)\in\T_{UX_{\K}Y,\gamma'}}.
\]
From Lemma~\ref{lem:typical-trans}, we have the fact that
\[
  \G(\uu,\yy)\subset\T_{X_{\K}|UY,\gamma'}(\uu,\yy)
\]
and $\xx_{\K}\in\G(\uu,\yy)$ implies $(\uu,\yy)\in\T_{UY,\gamma'}$.
Then, from Lemma~\ref{lem:typical-number-bound}, we have
\begin{align}
  |\G_{\K|\K^c}(\uu,\yy)|
  &\equiv
  |\G(\uu,\yy)|
  \notag
  \\*
  &\leq 
  |\T_{X_{\K}|UY,\gamma'}(\uu,\yy)|
  \notag
  \\
  &\leq
  2^{n[H(X_{\K}|UY)+\eta_{\X_{\K}|\U\Y}(\gamma'|\gamma')]}.
  \label{eq:typenumber-K}
\end{align}
For each non-empty set $\J\subsetneq\K$, let
\begin{align*}
  \G_{\X_{\J^c}}(\uu,\yy)
  &\equiv
  \lrb{
    \xx_{\J^c}: \xx_{\K}\in\G(\uu,\yy)
    \ \text{for some}\ \xx_{\J}\in\X_{\J}^n
  }
  \\
  \G_{\X_{\J}|\X_{\J^c}}(\uu,\xx_{\J^c},\yy)
  &\equiv
  \lrb{\xx_{\J}: \xx_{\K}\in\G(\uu,\yy)}.
\end{align*}
Then, from Lemma~\ref{lem:typical-trans}, we have the fact that
$\xx_{\J^c}\in\G_{\X_{\J^c}}(\uu,\yy)$ implies
$(\uu,\xx_{\J^c},\yy)\in\T_{UX_{\J^c}Y,\gamma'}$ and
\[
  \G_{\X_{\J}|\X_{\J^c}}(\uu,\xx_{\J^c},\yy)
  \subset\T_{X_{\J}|UX_{\J^c}Y,\gamma'}\lrsb{\uu,\xx_{\J^c},\yy}
\]
for every non-empty set $\J\subsetneq\K$. We have
\begin{align}
  |\G_{\J|\J^c}(\uu,\yy)|
  &\equiv
  \max_{\xx_{\J^c}\in\G_{\X_{\J^c}}(\uu,\yy)}
  \lrbar{\G_{\X_{\J}|\X_{\J^c}}\lrsb{\uu,\xx_{\J^c},\yy}}
  \notag
  \\*
  &\leq
  \max_{(\uu,\xx_{\J^c},\yy)\in\T_{UX_{\J^c}Y,\gamma'}}
  \lrbar{\T_{X_{\J}|UX_{\J^c}Y,\gamma'}\lrsb{\uu,\xx_{\J^c},\yy}}
  \notag
  \\
  &\leq
  2^{n[H(X_{\J}|U,X_{\J^c},Y)+\eta_{\X_{\J}|\U\X_{\J^c}\Y}(\gamma'|\gamma')]}
  \notag
  \\
  &\leq
  2^{n\lrB{H(X_{\J}|U,X_{\J^c},Y)+\eta_{\X_{\K}|\U\Y}(\gamma'|\gamma')}}
  \label{eq:typenumberj}
\end{align}
for every non-empty set $\J\subsetneq\K$, where the second inequality
comes from Lemma~\ref{lem:typical-number-bound}. Then, from
(\ref{eq:typenumber-K}), (\ref{eq:typenumberj}), and
Lemma~\ref{lem:multi-CRP}, we have
\begin{align}
  E_{\sfA_{\K}}\lrB{
    \chi(\hg_{\sfA_{\K}}(\sfA_{\K}\xx_{\K}|\yy,\uu)\neq\xx_{\K})
  }
  &\leq
  p_{\sfA_{\K}}\lrsb{\lrb{
      A_{\K}:
      \lrB{\G(\uu,\yy)\setminus\{\xx_{\K}\}}\cap\C_{A_{\K}}(A_{\K}\xx_{\K})
      \neq\emptyset
  }}
  \notag
  \\
  &\leq
  \sum_{\substack{
      \J\subset\K\\
      \J\neq\emptyset
  }}
  \frac{
    2^{n\lrB{H(X_{\J}|UX_{\J^c},Y)+\eta_{\X_{\K}|\U\Y}(\gamma'|\gamma')}}
    \alpha_{\sfA_{\J}}\lrB{\beta_{\sfA_{\J^c}}+1}
  }
  {\prod_{j\in\J}|\im\A_j|}
  +\beta_{\sfA_{\K}}
  \label{eq:CRP}
\end{align}
for all $(\uu,\xx_{\K},\yy)\in\T_{UX_{\K}Y,\gamma'}$. Then we have
\begin{align}
  &
  E_{\sfA_{\K}\sfA'_{\K}\sfaa_{\K}}\lrB{
    \sum_{\uu\in\T_{U,\gamma}}\mu_U(\uu)
    p_{M_{\K}Y}([\cap_{j=1}^{k+2}\cS_j]\cap\cS_{k+3}^c)
  }
  \notag
  \\*
  &\leq
  E_{\sfA_{\K}\sfaa_{\K}M_{\K}}
  \left[
    \sum_{\uu\in\T_{U,\gamma}}\mu_U(\uu)
    \sum_{\xx_{\K}\in\T_{\K}(\uu)}
    \lrB{\prod_{j\in\K}\chi(\hg_{\sfA_j\sfA'_j}(\sfaa_j,M_j|\uu)=\xx_j)}
    \vphantom{
      \sum_{\yy\in\T_{Y|X_{\K},\gamma}(\xx_{\K})}
      \mu_{Y|X_{\K}}(\yy|\xx_{\K})
      \chi(\hg_{\sfA_{\K}}(\sfaa_{\K}|\yy,\uu)\neq\xx_{\K})
    }
  \right.
  \notag
  \\*
  &\qquad\qquad\qquad\qquad\qquad\qquad\qquad\qquad\qquad\qquad\qquad
  \left.
    \vphantom{
      \sum_{\uu\in\T_{U,\gamma}}\mu_U(\uu)
      \sum_{\xx_{\K}\in\T_{\K}(\uu)}
      \lrB{\prod_{j\in\K}\chi(\hg_{\sfA_j\sfA'_j}(\sfaa_j,M_j|\uu)=\xx_j)}
    }
    \sum_{\yy\in\T_{Y|X_{\K},\gamma}(\xx_{\K})}
    \mu_{Y|X_{\K}}(\yy|\xx_{\K})
    \chi(\hg_{\sfA_{\K}}(\sfaa_{\K}|\yy,\uu)\neq\xx_{\K})
  \right]
  \notag
  \\*
  &\leq
  \sum_{\substack{
      \uu\in\T_{U,\gamma}\\
      \xx_{\K}\in\T_{\K}(\uu)\\
      \yy\in\T_{Y|X_{\K},\gamma}(\xx_{\K})
  }}
  \mu_U(\uu)\mu_{Y|X_{\K}}(\yy|\xx_{\K})
  \notag
  \\*
  &\qquad\qquad\qquad\qquad\qquad\cdot
  E_{\sfA_{\K}}\left[
    \chi(\hg_{\sfA_{\K}}(\sfA_{\K}\xx_{\K}|\yy,\uu)\neq\xx_{\K})
    \prod_{j\in\K}
    E_{\sfaa_j}\lrB{
      \chi(\sfA_j\xx_j=\sfaa_j)
    }
    E_{\sfA'_jM_j}\lrB{
      \chi(\sfA'_j\xx_j=M_j)
    }
  \right]
  \notag
  \\
  &=
  \frac 1{\prod_{j\in\K}|\im\hA_j|}
  \sum_{\substack{
      \uu\in\T_{U,\gamma}\\
      \xx_{\K}\in\T_{\K}(\uu)\\
      \yy\in\T_{Y|X_{\K},\gamma}(\xx_{\K})
  }}
  \mu_U(\uu)
  \mu_{Y|X_{\K}}(\yy|\xx_{\K})
  E_{\sfA_{\K}}\lrB{
    \chi(\hg_{\sfA_{\K}}(\sfA_{\K}\xx_{\K}|\yy,\uu)\neq\xx_{\K})
  }
  \notag
  \\
  &\leq
  \frac 1{\prod_{j\in\K}|\im\hA_j|}
  \sum_{\substack{
      \uu\in\T_{U,\gamma}\\
      \xx_{\K}\in\T_{\K}(\uu)\\
      \yy\in\T_{Y|X_{\K},\gamma}(\xx_{\K})
  }}
  \mu_U(\uu)
  \mu_{Y|X_{\K}}(\yy|\xx_{\K})
  \notag
  \\*
  &\qquad\qquad\qquad\qquad\qquad\qquad\qquad\qquad\qquad\cdot
  \lrB{
    \sum_{\substack{
	\J\subset\K\\
	\J\neq\emptyset
    }}
    \frac{
      2^{n\lrB{H(X_{\J}|U,X_{\J^c},Y)+\eta_{\X_{\K}|\U\Y}(\gamma'|\gamma')}}
      \alpha_{\sfA_{\J}}[\beta_{\sfA_{\J^c}}+1]
    }
    {\prod_{j\in\J}|\im\A_j|}
    +\beta_{\sfA_{\K}}
  }
  \notag
  \\
  &\leq
  2^k\kappa^k
  \lrB{
    \sum_{\substack{
	\J\subset\K\\
	\J\neq\emptyset
    }}
    2^{-n\lrB{\sum_{j\in\J}r_j-H(X_{\J}|U,X_{\J^c},Y)-\eta_{\X_{\K}|\U\Y}(\gamma'|\gamma')}}
    \alpha_{\sfA_{\J}}[\beta_{\sfA_{\J^c}}+1]
    +\beta_{\sfA_{\K}}
  }
  \notag
  \\
  &\leq
  \frac{\delta}{k+4}
  \label{eq:error[k+3]}
\end{align}
for all $\delta>0$ and all sufficiently large $n$, where $\T_{\K}(\uu)$
is defined as
\[
  \T_{\K}(\uu)\equiv\prod_{j\in\K}\T_j(\uu),
\]
the equality comes from Lemma~\ref{lem:E} that appears in
Appendix~\ref{sec:lemE}, the third inequality comes from (\ref{eq:CRP}),
the fourth inequality comes from (\ref{eq:rj}) and (\ref{eq:Tj}), and
the last inequality comes from (\ref{eq:multi-alpha}),
(\ref{eq:multi-beta}), (\ref{eq:rj-CRP-ts}), and (\ref{eq:k2}).

Finally, from (\ref{eq:error-ave})--(\ref{eq:errorj}),
(\ref{eq:error[k+1]}), (\ref{eq:error[k+2]}), and (\ref{eq:error[k+3]}),
we have the fact that for all $\delta>0$ and all sufficiently large $n$
there are $\{A_j,A'_j,\ba_j\}_{j\in\K}$, and $\uu$ satisfying
$A_j\in\A$, $A_j'\in\A_j'$, $\ba_j\in\im\A_j$, $\uu\in\U^n$ and
(\ref{eq:error-ts}).
\hfill\QED

\subsection{Proof of Theorem \ref{thm:sw}}

We can prove the theorem similarly to the proof of Theorem~\ref{thm:mac-ts}.

In the following, we assume that ensembles $(\bcA_j,\bpA{j})$ and
$(\bcAp_j,\bpAp{j})$ have a hash property for all $j\in\tK$. Similarly
to the proof of Theorem~\ref{thm:mac-ts}, we define an ensemble
$(\bhcA_j,\bphA{j})$ and $(\aalphahA{j},\bbetahA{j})$ for each
$j\in\tK$. Then we have the fact that $(\bhcA_j,\bphA{j})$ has a
$(\aalphahA{j},\bbetahA{j})$-hash property and there is a sequence
$\kkappa\equiv\{\kappa(n)\}_{n=1}^{\infty}$ such that
\begin{gather}
  \limn\kappa(n)=\infty
  \label{eq:k1-sw}
  \\
  \limn [\kappa(n)]^3\beta_{\sfA_{\tK}}(n)=0
  \label{eq:k2-sw}
  \\
  \limn\frac{\log\kappa(n)}n=0.
  \label{eq:k3-sw}
\end{gather}
From (\ref{eq:k3-sw}), we have the fact that there is a $\gamma$ such
that $\gamma>0$ and
\begin{gather}
  \eta_{\X_0}(\gamma)+\frac{\log\kappa}n
  \leq\e_0
  \label{eq:ekappa-sw}
  \\
  \eta_{\X_j|\X_0}(\gamma|\gamma)+\frac{\log\kappa}n
  \leq\e_j
  \label{eq:ekappaj-sw}
\end{gather}
for all $j\in\tK$ and all sufficiently large $n$ and
\begin{equation}
  5\gamma+\sum_{j\in\tK}\iota_j(\gamma)
  \leq\sum_{j\in\tK}\e_j,
  \label{eq:e0e1e2e-sw}
\end{equation}
where $\iota_j(\gamma)$ is defined by
\[
  \iota_j(\gamma)\equiv
  \begin{cases}
    \iota_{\X_0}(\gamma) &\text{if}\ j=0
    \\
    \iota_{\X_j|\X_0}(\gamma|\gamma) &\text{if}\ j\in\K.
  \end{cases}
\]
Similarly to the proof of (\ref{eq:Tj}), from (\ref{eq:ekappa-sw}),
we have the fact that there is a set $\T_0$ such that
$\T_0\subset\T_{X_0,\gamma}$  and
\begin{equation}
  \kappa
  \leq
  \frac{|\T_0|}{|\im\hcA_0|}
  \leq 
  2\kappa.
  \label{eq:T0-sw}
\end{equation}
We assume that $\T_0$ is constructed by selecting $|\T_0|$ elements in
the ascending order regarding the value $D(\nu_{\xx_0}\|\mu_{X_0})$.
Furthermore, from (\ref{eq:ekappaj-sw}), we have the fact that for all
$\xx_0\in\X_0^n$  and all $j\in\K$ there is a set $\T_j(\xx_0)$ such
that $\T_j(\xx_0)\subset\T_{X_j|X_0,\gamma}(\xx_0)$ and
\begin{equation}
  \kappa
  \leq
  \frac{|\T_j(\xx_0)|}{|\im\hcA_1|}
  \leq 
  2\kappa.
  \label{eq:Tj-sw}
\end{equation}
We assume that $\T_j(\xx_0)$ is constructed by selecting $|\T_j(\xx_0)|$
elements in the ascending order regarding the value
$D(\nu_{\xx_j|\xx_0}\|\mu_{X_j|X_0}|\nu_{\xx_0})$.

Now we prove the theorem. Let $\mm_0\in\M_0$ be a common message and
$\mm_1\in\M_1$ and $\mm_2\in\M_2$ be private messages. Let
$\xx_0\in\X_0^n$ be defined as
\[
  \xx_0\equiv\hg_{A_0A'_0}(\ba_0,\mm_0).
\]
Let $(\xx_1,\xx_2)$ be channel inputs, where $\xx_j\in\X_j^n$ is defined
by
\[
  \xx_j
  \equiv \hg_{A_jA'_j}(\ba_j,\mm_j|\xx_0)
  \quad\text{for each}\ j\in\K.
\]
Let $\yy\in\Y^n$ be the channel output. We define
\begin{align*}
  \cS_0
  &\equiv
  \lrb{(\mm_{\tK},\yy):
    \xx_0\in\T_0\subset\T_{X_0,\gamma}
    \ \text{and}\ \yy\in\Y^n
  }
  \\
  \cS_j
  &\equiv
  \lrb{(\mm_{\tK},\yy):
    \xx_j\in\T_j(\xx_0)\subset\T_{X_j|X_0,\gamma}(\xx_0)
    \ \text{and}\ \yy\in\Y^n
  }
  \\
  \cS_3
  &\equiv
  \lrb{(\mm_{\tK},\yy):
    I(\xx_1;\xx_2|\xx_0)<\gamma+\sum_{j\in\tK}\lrB{\iota_{j}(\gamma)+\e_j}
    \ \text{and}\ \yy\in\Y^n 
  }
  \\
  \cS_4
  &\equiv
  \lrb{(\mm_{\tK},\yy):
    \yy\in\T_{Y|X_{\tK},\gamma}(\xx_{\tK})
  }
  \\
  \cS_5
  &\equiv
  \lrb{(\mm_{\tK},\yy):
    \hg_{A_{\tK}}(\ba_{\tK}|\yy)=\xx_{\tK}
  },
\end{align*}
where $j\in\K$ and
\begin{align}
  I(\xx_1;\xx_2|\xx_0)
  &\equiv 
  \sum_{j\in\{1,2\}}H(\xx_j|\xx_0)-H(\xx_{\{1,2\}}|\xx_0)
  \notag
  \\
  &=
  H(\xx_1|\xx_0)+H(\xx_2|\xx_0)-H(\xx_1,\xx_2|\xx_0).
\end{align}
The error probability is upper bounded by
\begin{align}
  &
  \Error(A_{\tK},A'_{\tK},\ba_{\tK})
  \notag
  \\*
  &
  \leq
  p_{M_{\tK}Y}(\cS_0^c)+
  \sum_{j\in\K}p_{M_{\tK}Y}(\cS_0\cap\cS_j^c)
  +p_{M_{\tK}Y}\lrsb{\lrB{\cap_{j=0}^2\cS_j}\cap\cS_3^c}
  +p_{M_{\tK}Y}(\cS_4^c)
  +p_{M_{\tK}Y}\lrsb{\lrB{\cap_{j=0}^4\cS_j}\cap\cS_{5}^c},
  \label{eq:error-sum-sw}
\end{align}
We assign equation numbers to the conditions
\begin{gather}
  \xx_0\in\T_0\subset\T_{X_0,\gamma}
  \label{eq:MAC-SW0}
  \\
  \xx_j\in\T_j(\xx_0)\subset\T_{X_j|X_0,\gamma}(\xx_0)
  \quad\text{for all}\ j\in\K\equiv{1,2}
  \label{eq:MAC-SWj}
  \\
  \label{eq:MAC-SW3}
  I(\xx_1;\xx_2|\xx_0)<\gamma+\sum_{j\in\tK}\lrB{\iota_{j}(\gamma)+\e_j}
  \\
  \label{eq:MAC-SW4}
  \yy\in\T_{Y|X_{\tK},\gamma}(\xx_{\tK})
  \\
  \hg_{A_{\tK}}(\ba_{\tK}|\yy)=\xx_{\tK}
  \label{eq:MAC-SW5}
\end{gather}
which are referred later. In comparison with the conventional
superposition coding, the condition (\ref{eq:MAC-SW0}) corresponds to an
event where the function $\hg_{A_0A_0'}$ finds a `good' cloud center
$\xx_0$ and the condition (\ref{eq:MAC-SWj}) corresponds to an event
where the function $\hg_{A_jA'_j}$ finds a `good' satellite $\xx_j$
for all $j\in\K$, where `good' means that they are (conditionally)
typical sequences. When (\ref{eq:MAC-SW0})--(\ref{eq:MAC-SW4}) are
satisfied, we have the fact that $(\xx_0,\xx_1,\xx_2)$ is jointly
typical. It should be noted that (\ref{eq:MAC-SW3}) was unnecessary
in the proof of the conventional superposition coding because the joint
typicality of $(\xx_0,\xx_1,\xx_2)$ was naturally satisfied by
generating codewords at random.

In the following, let
\[
  \haa_j\equiv (\ba_j,\mm_j)
  \quad\text{for each}\  j\in\tK.
\]
We assume that the distribution of $\haa_j$ is uniform on $\im\hcA_j$
for all $j\in\tK$, and $\{\sfhA_j,\sfaa_j,M_j\}_{j\in\tK}$ are mutually
independent.

First, we evaluate $E_{\sfhA_{\tK}\sfaa_{\tK}}\lrB{p_{M_{\tK}Y}(\cS_0^c)}$.
From Lemma~\ref{lem:SP} and (\ref{eq:T0-sw}), we have
\begin{align}
  E_{\sfhA_{\tK}\sfaa_{\tK}}\lrB{p_{M_{\tK}Y}(\cS_0^c)}
  &=
  p_{\sfhA_0\sfaa_0M_0}
  \lrsb{\lrb{
      (A_0,A'_0,\ba_0,\mm_0):
      \hg_{A_0A'_0}(\ba_0,\mm_0)\notin\T_0
  }}
  \notag
  \\
  &\leq
  p_{\sfhA_0\sfhaa_0}
  \lrsb{\lrb{
      (\hA_0,\haa_0):
      \T_0\cap\C_{\hA_0}(\haa_0)=\emptyset
  }}
  \notag
  \\
  &\leq
  \alpha_{\sfhA_0}-1
  +\frac{|\im\hcA_0|\lrB{\beta_{\sfhA_0}+1}}
  {|\T_0|}
  \notag
  \\
  &\leq
  \alpha_{\sfhA_0}-1
  +\frac{\beta_{\sfhA_0}+1}
  {\kappa}
  \notag
  \\
  &\leq
  \frac {\delta}6
  \label{eq:error0-sw}
\end{align} 
for all $\delta>0$ and sufficiently large $n$, where the last inequality
comes from (\ref{eq:k1-sw}) and the fact that $\alphahA{0}\to1$ and
$\betahA{0}\to0$ as $n\to\infty$.

Next, we evaluate
$E_{\sfhA_{\tK}\sfaa_{\tK}}\lrB{p_{M_{\tK}Y}(\cS_0\cap\cS_j^c)}$.
From Lemma~\ref{lem:SP} and (\ref{eq:Tj-sw}), we have
\begin{align}
  &
  E_{\sfhA_{\tK}\sfaa_{\tK}}\lrB{p_{M_{\tK}Y}(\cS_0\cap\cS_j^c)}
  \notag
  \\*
  &=
  \sum_{\hA_0,\sfhaa_0}
  p_{\sfhA_0\sfhaa_0}(\hA_0,\haa_0)
  \sum_{\xx_0\in\T_0}
  \chi\lrsb{\hg_{\hA_0}(\haa_0)=\xx_0}
  p_{\sfhA_j\sfhaa_j}
  \lrsb{\lrb{
      (\hA_j,\haa_j):
      \hg_{\hA_j}(\haa_j|\xx_0)\notin\T_j(\xx_0)
  }}
  \notag
  \\
  &\leq
  \sum_{\hA_0,\sfhaa_0}
  p_{\sfhA_0\sfhaa_0}(\hA_0,\haa_0)
  \sum_{\xx_0\in\T_0}
  \chi\lrsb{\hg_{\hA_0}(\haa_0)=\xx_0}
  p_{\sfhA_j\sfhaa_j}
  \lrsb{\lrb{
      (\hA_j,\haa_j):
      \T_j(\xx_0)\cap\C_{\hA_j}(\haa_j)=\emptyset
  }}
  \notag
  \\
  &\leq
  \sum_{\hA_0,\sfhaa_0}
  p_{\sfhA_0\sfhaa_0}(\hA_0,\haa_0)
  \sum_{\xx_0\in\T_0}
  \chi\lrsb{\hg_{\hA_0}(\haa_0)=\xx_0}
  \lrB{
    \alpha_{\sfhA_j}-1
    +\frac{|\im\hcA_j|\lrB{\beta_{\sfhA_j}+1}}
    {|\T_j(\xx_0)|}
  }
  \notag
  \\
  &\leq
  \alpha_{\sfhA_j}-1
  +\frac{\beta_{\sfhA_j}+1}
  {\kappa}
  \notag
  \\
  &\leq
  \frac {\delta}6
  \label{eq:errorj-sw}
\end{align} 
for all $\delta>0$ and sufficiently large $n$, where the last inequality
comes from (\ref{eq:k1-sw}) and the fact that $\alphahA{j}\to1$ and
$\betahA{j}\to0$ as $n\to\infty$.

Next, we evaluate
$E_{\sfhA_{\tK}\sfaa_{\tK}}\lrB{
  p_{M_{\tK}Y}\lrsb{\lrB{\cap_{j=0}^2\cS_j}\cap\cS_3^c}}$.
Assume that $(\xx_0,\xx_1,\xx_2)$ satisfies (\ref{eq:MAC-SW0}),
(\ref{eq:MAC-SWj}) and
\[
  I(\xx_1;\xx_2|\xx_0)
  \geq
  \gamma+\sum_{j\in\tK}\lrB{\iota_{j}(\gamma)+\e_j}.
\]
Then, from Lemma~\ref{lem:diff-entropy}, we have
\begin{gather*}
  |H(\xx_0)-H(X_0)|
  <\iota_{\X_0}(\gamma)
  \\
  |H(\xx_1|\xx_0)-H(X_1|X_0)|
  <\iota_{\X_1|\X_0}(\gamma|\gamma)
  \\
  |H(\xx_2|\xx_0)-H(X_2|X_0)|
  <\iota_{\X_2|\X_0}(\gamma|\gamma).
\end{gather*}
Then we have
\begin{align}
  H(\xx_0,\xx_1,\xx_2)
  &=H(\xx_0)+H(\xx_1|\xx_0)+H(\xx_2|\xx_0)
  -I(\xx_1;\xx_2|\xx_0)
  \notag
  \\
  &\leq
  H(X_0)+H(X_1|X_0)+H(X_2|X_0)
  -\lrB{\gamma+\sum_{j\in\tK}\lrB{\iota_{j}(\gamma)+\e_j}}
  \notag
  \\
  &=\sum_{j\in\tK}\lrB{r_j+R_j}-\gamma,
\end{align}
where the last equality comes from (\ref{eq:r0-sw})--(\ref{eq:r2-sw}).
Since $\xx_j\in\C_{A_jB_j}(\ba_j,\mm_j)$ for all $j\in\tK$, we have
\[
  (\xx_0,\xx_1,\xx_2)
  \in
  \G\cap\C_{\hA_{\tK}}(\haa_{\tK}),
\]
where $\G\subset\X_0^n\times\X_1^n\times\X_2^n$ is defined as
\[
  \G\equiv
  \lrb{
    (\xx_0,\xx_1,\xx_2):
    H(\xx_0,\xx_1,\xx_2)
    <\sum_{j\in\tK}[r_j+R_j]-\gamma
  }.
\]
This implies that
\[
  \G\cap\C_{\hA_{\tK}}(\haa_{\tK})
  \neq\emptyset.
\]
Similarly to the proof of (\ref{eq:error[k+1]}), we have
\begin{align}
  E_{\sfhA_{\tK}\sfaa_{\tK}}\lrB{
    p_{M_{\tK}Y}\lrsb{\lrB{\cap_{j=0}^2\cS_j}\cap\cS_3^c}
  }
  &\leq
  p_{\sfhA_{\tK}\sfaa_{\tK}}
  \lrsb{\lrb{
      (\hA_{\tK},\haa_{\tK}):
      \G\cap\C_{\hA_{\tK}}(\haa_{\tK})\neq\emptyset
  }}
  \notag
  \\
  &\leq
  \sum_{\xx_{\tK}\in\G}
  p_{\sfhA_{\tK}\sfaa_{\tK}}
  \lrsb{\lrb{
      (\hA_{\tK},\haa_{\tK}):
      \hA_j\xx_j=\haa_j
      \ \text{for all}\ j\in\tK
  }}
  \notag
  \\
  &=
  \frac {|\G|}{\prod_{j\in\tK}|\im\hcA_j|}
  \notag
  \\
  &\leq
  \frac{2^{n\lrB{\sum_{j\in\tK}[r_j+R_j]-\gamma+\lambda_{\X_{\tK}}}}}
  {\prod_{j\in\tK}|\im\hcA_j|}
  \notag
  \\*
  &=
  2^{-n[\gamma-\lambda_{\X_{\tK}}]}
  \notag
  \\
  &\leq\frac{\delta}6
  \label{eq:error3-sw}
\end{align} 
for all $\delta>0$ and all sufficiently large $n$.

Next, we evaluate
$E_{\sfhA_{\tK}\sfaa_{\tK}}\lrB{p_{M_{\tK}Y}(\cS_4^c)}$.
Similarly to the proof of (\ref{eq:error[k+2]}), we have
\begin{align}
  E_{\sfhA_{\tK}\sfaa_{\tK}}\lrB{p_{M_{\tK}Y}(\cS_4^c)}
  &
  =
  E_{\sfhA_{\tK}\sfhaa_{\tK}}\left[
    \mu_{Y|X_1X_2}\lrsb{
      \left.
	\lrB{\T_{Y|X_{\tK},\gamma}(\XX_{\tK})}^c
      \right|\XX_1,\XX_2
    }
  \right]
  \notag
  \\*
  &
  =
  E_{\sfhA_{\tK}\sfhaa_{\tK}}\left[
    \mu_{Y|X_{\tK}}\lrsb{
      \left.
	\lrB{\T_{Y|X_{\tK},\gamma}(\XX_{\tK})}^c
      \right|\XX_{\tK}
    }
  \right]
  \notag
  \\
  &\leq
  2^{-n[\gamma-\lambda_{\X_{\tK}\Y}]}
  \notag
  \\
  &
  \leq
  \frac {\delta}6
  \label{eq:error4-sw}
\end{align}
for all $\delta>0$ and all sufficiently large $n$, where
$\XX_{\tK}\equiv\{\XX_i\}_{i\in\tK}$ is defined by
\begin{gather*}
  \XX_0\equiv \hg_{A_0A'_0}(\sfaa_0,M_0)
  \\
  \XX_j\equiv \hg_{A_jA'_j}(\sfaa_j,M_j|\XX_0)
  \quad\text{for each}\ j\in\K.
\end{gather*}

Next, we evaluate
$E_{\sfhA_{\tK}\sfaa_{\tK}}\lrB{
  p_{M_{\tK}Y}\lrsb{\lrB{\cap_{j=0}^4\cS_j}\cap\cS_5^c}
}$.
In the following, we assume (\ref{eq:MAC-SW0})--(\ref{eq:MAC-SW4}) and
\[
  g_{A_0A_1A_2}(\ba_0,\ba_1,\ba_2|\yy)\neq(\xx_0,\xx_1,\xx_2).
\]
Similarly to the proof of (\ref{eq:D}), we have
\begin{align}
  D(\nu_{\xx_{\tK}\yy}\parallel\mu_{X_{\tK}Y})
  &=
  D(\nu_{\yy|\xx_{\tK}}\parallel\mu_{Y|X_{\tK}}|\nu_{\xx_{\tK}})
  +\sum_{j\in\{0,1\}}
  D(\nu_{\xx_j|\xx_0}\parallel\mu_{X_j|X_0}|\nu_{\xx_0})
  +D(\nu_{\xx_0}\parallel\mu_{X_0})
  +I(\xx_1;\xx_2|\xx_0)
  \notag
  \\
  &<5\gamma+\sum_{j\in\tK}\lrB{\iota_{j}(\gamma)+\e_j}
  \notag
  \\
  &
  \leq
  2\sum_{j\in\tK}\e_j,
  \label{eq:D012}
\end{align}
where the last inequality comes from (\ref{eq:e0e1e2e-sw}). This implies
that
\[
  (\xx_0,\xx_1,\xx_2,\yy)\in\T_{X_{\tK}Y,\gamma'}.
\]
where $\gamma'$ is defined as
\[
  \gamma'
  \equiv 2\sum_{j\in\tK}\e_j.
\]
Since $\hg_{A_0A_1A_2}(\ba_0,\ba_1,\ba_2|\yy)\neq(\xx_0,\xx_1,\xx_2)$,
there is $(\xx_0',\xx_1',\xx_2')\in\C_{A_{\tK}}(\ba_{\tK})$ such that
$(\xx_0',\xx_1',\xx_2')\neq(\xx_0,\xx_1,\xx_2)$ and
$(\xx_0',\xx_1',\xx_2',\yy)\in\T_{X_{\tK}Y,\gamma'}$. This implies that
\[
  \lrB{\G(\yy)\setminus\{\xx_{\tK}\}}
  \cap\C_{A_{\tK}}(A_{\tK}\xx_{\tK})
  \neq\emptyset,
\]
where 
\[
  \G(\yy)\equiv
  \lrb{(\xx_0,\xx_1,\xx_2):(\xx_0,\xx_1,\xx_2,\yy)\in\T_{X_{\tK}Y,\gamma'}}.
\]
From Lemma~\ref{lem:typical-trans}, we have the fact that
\[
  \G(\yy)\subset\T_{X_{\tK}|Y,\gamma'}(\yy)
\]
and $(\xx_0,\xx_1,\xx_2)\in\G(\yy)$ implies $\yy\in\T_{Y,\gamma'}$.
Then, from Lemma~\ref{lem:typical-number-bound}, we have
\begin{align}
  |\G_{\tK|\tK^c}(\yy)|
  &\equiv
  |\G(\yy)|
  \notag
  \\
  &\leq 
  |\T_{X_{\tK}|Y,\gamma'}(\yy)|
  \notag
  \\
  &\leq
  2^{n[H(X_{\tK}|Y)+\eta_{\X_{\tK}|\Y}(\gamma'|\gamma')]}.
  \label{eq:typenumber-K-sw}
\end{align}
For each non-empty set $\J\subsetneq\tK$, let
\begin{align*}
  \G_{\X_{\J^c}}(\yy)
  &\equiv
  \lrb{
    \xx_{\J^c}:
    \xx_{\tK}\in\G(\yy)
    \ \text{for some}\ \xx_{\J}\in\X_{\J}^n
  }
  \\
  \G_{\X_{\J}|\X_{\J^c}}(\xx_{\J^c},\yy)
  &\equiv
  \lrb{\xx_{\J}:
    \xx_{\tK}\in\G(\yy)
  }.
\end{align*}
Then, from Lemma~\ref{lem:typical-trans}, we have the fact that
$\xx_{\J^c}\in\G_{\X_{\J^c}}(\yy)$ implies
$(\xx_{\J^c},\yy)\in\T_{X_{\J^c}Y,\gamma'}$ and
\[
  \G_{\X_{\J}|\X_{\J^c}}(\xx_{\J^c},\yy)
  \subset\T_{X_{\J}|X_{\J^c}Y,\gamma'}\lrsb{\xx_{\J^c},\yy}
\]
for every non-empty set $\J\subsetneq\tK$. We have
\begin{align}
  |\G_{\J|\J^c}(\yy)|
  &\equiv
  \max_{\xx_{\J^c}\in\G_{\X_{\J^c}}(\yy)}
  \lrbar{\G_{\X_{\J}|\X_{\J^c}}(\xx_{\J^c},\yy)}
  \notag
  \\
  &\leq
  \max_{(\xx_{\J^c},\yy)\in\T_{X_{\J^c}Y,\gamma'}}
  \lrbar{\T_{X_{\J}|X_{\J^c}Y,\gamma'}(\xx_{\J^c},\yy)}
  \notag
  \\
  &\leq
  2^{n[H(X_{\J}|X_{\J^c},Y)+\eta_{\X_{\J}|\X_{\J^c}\Y}(\gamma'|\gamma')]}
  \notag
  \\
  &\leq
  2^{n[H(X_{\J}|X_{\J^c},Y)+\eta_{\X_{\tK}|\Y}(\gamma'|\gamma')]}
  \label{eq:typenumber-sw}
\end{align}
for every non-empty set $\J\subsetneq\tK$, where the second inequality
comes from Lemma~\ref{lem:typical-number-bound}. Then, from 
(\ref{eq:typenumber-K-sw}), (\ref{eq:typenumber-sw}), and
Lemma~\ref{lem:multi-CRP}, we have
\begin{align}
  E_{\sfA_{\tK}}\lrB{
    \chi(\hg_{\sfA_{\tK}}(\sfA_{\tK}\xx_{\tK}|\yy)\neq\xx_{\tK})
  }
  &\leq
  p_{\sfA_{\tK}}\lrsb{\lrb{
      A_{\tK}:
      \lrB{\G(\yy)\setminus\{\xx_{\tK}\}}
      \cap\C_{A_{\tK}}(A_{\tK}\xx_{\tK})
      \neq\emptyset
  }}
  \notag
  \\
  &\leq
  \sum_{\substack{
      \J\subset\tK\\
      \J\neq\emptyset
  }}
  \frac{
    2^{n\lrB{H(X_{\J}|X_{\J^c},Y)+\eta_{\X_{\tK}|\Y}(\gamma'|\gamma')}}
    \alpha_{\sfA_{\J}}\lrB{\beta_{\sfA_{\J^c}}+1}
  }
  {\prod_{j\in\J}|\im\A_j|}
  +\beta_{\sfA_{\tK}}
  \label{eq:sw-CRP}
\end{align}
for all $(\xx_{\tK},\yy)\in\T_{X_{\tK}Y,\gamma'}$. Then we have
\begin{align}
  &
  E_{\sfhA_{\tK}\sfaa_{\tK}}\lrB{
    p_{M_{\tK}Y}\lrsb{\lrB{\cap_{j=0}^4\cS_j}\cap\cS_5^c}
  }
  \notag
  \\*
  &\leq
  E_{\sfhA_{\tK}\sfhaa_{\tK}}
  \left[
    \sum_{\xx_{\tK}\in\T}
    \chi(\hg_{\sfhA_0}(\sfhaa_0)=\xx_0)
    \lrB{\prod_{j\in\K}\chi(\hg_{\sfhA_j}(\sfhaa_j|\xx_0)=\xx_j)}
    \sum_{\yy\in\T_{Y|X_{\tK},\gamma}(\xx_{\tK})}
    \mu_{Y|X_{\tK}}(\yy|\xx_{\tK})
    \chi(\hg_{\sfA_{\tK}}(\sfaa_{\tK}|\yy)\neq\xx_{\tK})
  \right]
  \notag
  \\
  &\leq
  \sum_{\substack{
      \xx_{\tK}\in\T\\
      \yy\in\T_{Y|X_{\tK},\gamma}(\xx_{\tK})
  }}
  \mu_{Y|X_{\tK}}(\yy|\xx_{\tK})
  E_{\sfA_{\tK}}\lrB{
    \chi(\hg_{\sfA_{\tK}}(\sfA_{\tK}\xx_{\tK}|\yy)\neq\xx_{\tK})
    \prod_{j\in\tK}
    E_{\sfaa_j}\lrB{
      \chi(\sfA_j\xx_j=\sfaa_j)
    }
    E_{\sfA'_jM_j}
    \lrB{
      \chi(\sfA'_j\xx_j=M_j)
    }
  }
  \notag
  \\
  &=
  \frac 1{\prod_{j\in\tK}|\im\hA_j|}
  \sum_{\substack{
      \xx_{\tK}\in\T\\
      \yy\in\T_{Y|X_{\tK},\gamma}(\xx_{\tK})
  }}
  \mu_{Y|X_{\tK}}(\yy|\xx_{\tK})
  E_{\sfA_{\tK}}\lrB{
    \chi(\hg_{\sfA_{\tK}}(\sfA_{\tK}\xx_{\tK}|\yy)\neq\xx_{\tK})
  }
  \notag
  \\
  &\leq
  \frac 1{\prod_{j\in\tK}|\im\hA_j|}
  \sum_{\substack{
      \xx_{\tK}\in\T\\
      \yy\in\T_{Y|X_{\tK},\gamma}(\xx_{\tK})
  }}
  \mu_{Y|X_{\tK}}(\yy|\xx_{\tK})
  \lrB{
    \sum_{\substack{
	\J\subset\tK\\
	\J\neq\emptyset
    }}
    \frac{
      2^{n\lrB{H(X_{\J}|X_{\J^c},Y)+\eta_{\X_{\tK}|\Y}(\gamma'|\gamma')}}
      \alpha_{\sfA_{\J}}\lrB{\beta_{\sfA_{\J^c}}+1}
    }
    {\prod_{j\in\J}|\im\A_j|}
    +\beta_{\sfA_{\tK}}
  }
  \notag
  \\
  &\leq
  8\kappa^3
  \lrB{
    \sum_{\substack{
	\J\subset\tK\\
	\J\neq\emptyset
    }}
    2^{-n\lrB{\sum_{j\in\J}r_j-H(X_{\J}|X_{\J^c},Y)
	-\eta_{\X_{\tK}|\Y}(\gamma'|\gamma')}}
    \alpha_{\sfA_{\J}}\lrB{\beta_{\sfA_{\J^c}}+1}
    +\beta_{\sfA_{\tK}}
  }
  \notag
  \\
  &\leq
  \frac{\delta}6
  \label{eq:error5-sw}
\end{align}
for all $\delta>0$ and all sufficiently large $n$, where $\T$ is defined
as
\[
  \T\equiv\lrb{(\xx_0,\xx_1,\xx_2):
    \xx_0\in\T_0,
    \xx_1\in\T_1(\xx_0),
    \xx_2\in\T_2(\xx_0)
  },
\]
the equality comes from Lemma~\ref{lem:E}, which appears in
Appendix~\ref{sec:lemE}, the third inequality comes from
(\ref{eq:sw-CRP}), the fourth inequality comes from (\ref{eq:rj-sw}),
(\ref{eq:T0-sw}), and (\ref{eq:Tj-sw}), and the last inequality comes
from (\ref{eq:multi-alpha}), (\ref{eq:multi-beta}),
(\ref{eq:r0-sw-CRP})--(\ref{eq:r0r1r2-sw-CRP}), and (\ref{eq:k2-sw}).

Finally, from (\ref{eq:error-sum-sw})--(\ref{eq:errorj-sw}),
(\ref{eq:error3-sw}), (\ref{eq:error4-sw}), and (\ref{eq:error5-sw}),
we have the fact that for all $\delta>0$ and sufficiently large $n$
there are $\{A_j,A_j',\ba_j\}_{j\in\tK}$ satisfying $A_j\in\A$,
$A_j'\in\A_j'$, $\ba_j\in\im\A_j$, and (\ref{eq:error-sw}).
\hfill\QED

\appendix
\subsection{Basic Property of Ensemble}
\label{sec:lemE}

\begin{lem}[{\cite[Lemma 9]{HASH}}]
\label{lem:E}
Assume that random variables $\sfA$ and $\sfaa$ are independent and the
distribution of $\sfaa$ is uniform on $\im\A$. Then,
\[
  E_{\sfaa}\lrB{\chi(A\uu=\sfaa)}
  =\frac 1{|\im\A|}
\]
for any $A\in\A$ and $\uu\in\U^n$, and
\[
  E_{\sfA\sfaa}\lrB{\chi(\sfA\uu=\sfaa)}
  =\frac 1{|\im\A|}
\]
for any $\uu\in\U^n$.
\end{lem}

\subsection{Method of Types}
\label{sec:type-theory}

Let $\T_U\subset\U^n$ be a set of all sequences that has the same type
$\nu_U$, where type of $\uu\in\U^n$ is defined by the empirical
distribution $\nu_{\uu}$. Let $\T_{U,\gamma}$ be a set of typical
sequences and $\T_{U|V,\gamma}(\vv)$ be a set of conditionally typical
sequences defined in the beginning of Section~\ref{sec:def}.

\begin{lem}[{\cite[Lemma 2.2]{CK}}]
\label{lem:typebound}
The number of different types of sequences in $\U^n$ is fewer than
$[n+1]^{|\U|}$. The number of conditional types of sequences in
$\U^n\times\V^n$ is fewer than $[n+1]^{|\U||\V|}$.
\end{lem}

\begin{lem}[{\cite[Lemma 2.3 and 2.5]{CK}}]
\label{lem:typenumber}
Let $\lambda_{\U}$ be defined in (\ref{eq:lambda}). Then
\[
  2^{n[H(U)-\lambda_{\U}]}\leq |\T_U|\leq 2^{nH(U)}.
\]
\end{lem}

\begin{lem}
\label{lem:typical-number}
For $H\geq 0$, 
\[
  |\lrb{\uu: H(\uu)\leq H}|
  \leq
  2^{n[H+\lambda_{\U}]}
\]
where $\lambda_{\U}$ is defined by (\ref{eq:lambda}).
\end{lem}
\begin{proof}
The proof is similar to that of \cite[Lemma 6]{HASH-UNIV}. We have
\begin{align}
  |\lrb{\uu: H(\uu)\leq H}|
  &=
  \sum_{\substack{
      U:
      H(U)\leq H
  }}
  |\T_{U}|
  \notag
  \\
  &\leq
  \sum_{\substack{
      U:
      H(U)\leq H
  }}
  2^{nH(U)}
  \notag
  \\
  &\leq
  \sum_{\substack{
      U:
      H(U)\leq H
  }}
  2^{nH}
  \notag
  \\
  &\leq
  [n+1]^{|\U|}2^{nH}
  \notag
  \\
  &=
  2^{n[H+\lambda_{\U}]},
\end{align}
where the sum is taken over all random variables $U$ corresponding 
the type of a sequence in $\U^n$, the first inequality comes from
Lemma~\ref{lem:typenumber}, and the last inequality comes from
Lemma~\ref{lem:typebound}.
\end{proof}

\begin{lem}[{\cite[Lemma 22]{HASH}\cite[Theorem 2.5]{UYE}}]
\label{lem:typical-trans}
If $\vv\in\T_{V,\gamma}$ and $\uu\in\T_{U|V,\gamma'}(\vv)$, then
$(\uu,\vv)\in\T_{UV,\gamma+\gamma'}$.
If $(\uu,\vv)\in\T_{UV,\gamma}$, then $\uu\in\T_{U,\gamma}$ and
$\uu\in\T_{U|V,\gamma}(\vv)$.
\end{lem}

\begin{lem}[{\cite[Theorem 2.6]{UYE}}]
\label{lem:type}
If $\uu\in\T_{U,\gamma}$, then
\[
  \sum_{u\in\U}|\nu_{\uu}(u)-\mu_U(u)|\leq \sqrt{2\gamma}
\]
\end{lem}
\begin{proof}
The statement is shown by
\begin{align}
  \sum_{u\in\U}|\nu_{\uu}(u)-\mu_U(u)|
  &\leq\sqrt{\frac{2D(\nu_{\uu}\parallel\mu_U)}{\log_2 e}}
  \notag
  \\
  &\leq \sqrt{\frac{2\gamma}{\log_2 e}}
  \notag
  \\
  &\leq \sqrt{2\gamma},
\end{align}
where $e$ is the base of the natural logarithm and the first inequality
comes from \cite[Lemma 12.6.1]{CT}.
\end{proof}

\begin{lem}
\label{lem:diff-entropy}
Let $0<\gamma\leq 1/8$. If $\vv\in\T_{V,\gamma}$, and
 $\uu\in\T_{U|V,\gamma'}(\vv)$, then
\begin{gather*}
  |H(\vv)-H(V)|\leq\iota_{\V}(\gamma)
  \\
  |H(\uu|\vv)-H(U|V)|\leq\iota_{\U|\V}(\gamma'|\gamma),
\end{gather*}
where $\iota_{\U}$ and $\iota_{\U|\V}$ are defined by (\ref{eq:iota})
and (\ref{eq:iotac}), respectively.
\end{lem}
\begin{proof}
From \cite[Lemma 2.7]{CK}, we have
\begin{equation}
  |H(p)-H(q)|\leq -\theta\log\frac{\theta}{|\U|}
  \label{eq:diff-H}
\end{equation}
for any $\theta$ and probability distributions $p$ and $q$ on $\V$
 satisfying
\[
  \sum_{v\in\V}|p(v)-q(v)|\leq \theta\leq \frac 12.
\]
Then the first inequality is shown by this fact and Lemma \ref{lem:type}.

Next we prove the second inequality. Let $\nu_{\uu|\vv}\nu_{\vv}$ and
$\mu_{U|V}\nu_{\vv}$ be defined as
\begin{align*}
  \nu_{\uu|\vv}\nu_{\vv}(u,v)
  &\equiv\nu_{\uu|\vv}(u|v)\nu_{\vv}(v)
  \\
  \mu_{U|V}\nu_{\vv}(u,v)
  &\equiv\mu_{U|V}(u|v)\nu_{\vv}(v),
\end{align*}
respectively. Since
\begin{align}
  D(\nu_{\uu|\vv}\nu_{\vv}\parallel\mu_{U|V}\nu_{\vv})
  &=
  \sum_{u,v}
  \nu_{\uu|\vv}(u|v)\nu_{\vv}(v)
  \log\frac{\nu_{\uu|\vv}(u|v)}{\mu_{U|V}(u|v)}
  \notag
  \\
  &=D(\nu_{\uu|\vv}\parallel\mu_{U|V}|\nu_{\vv})
  \notag
  \\
  &<\gamma'
\end{align}
we have
\begin{align}
  &|H(\nu_{\uu|\vv}|\nu_{\vv})-H(\mu_{U|V}|\nu_{\vv})|
  \notag
  \\*
  &=
  \left|
    \sum_{u,v}
    \nu_{\uu|\vv}(u|v)\nu_{\vv}(v)
    \log\frac1{\nu_{\uu|\vv}(u|v)}
    -\sum_{u,v}
    \mu_{U|V}(u|v)\nu_{\vv}(v)
    \log\frac1{\mu_{U|V}(u|v)}
  \right|
  \notag
  \\
  &=
  \left|
    \sum_{u,v}
    \nu_{\uu|\vv}(u|v)\nu_{\vv}(v)
    \log\frac{\nu_{\vv}(v)}{\nu_{\uu|\vv}(u|v)\nu_{\vv}(v)}
    -\sum_{u,v}
    \mu_{U|V}(u|v)\nu_{\vv}(v)
    \log\frac{\nu_{\vv}(v)}{\mu_{U|V}(u|v)\nu_{\vv}(v)}
  \right|
  \notag
  \\
  &=
  |H(\nu_{\uu|\vv}\nu_{\vv})-H(\mu_{U|V}\nu_{\vv})|
  \notag
  \\
  &\leq
  \iota_{\U\V}(\gamma'),
  \label{eq:diff-Hc}
\end{align}
where the last inequality comes from (\ref{eq:diff-H}). We have
\begin{align}
  |H(\nu_{\uu|\vv}|\nu_{\vv})-H(\nu_{\uu|\vv}|\mu_{V})|
  &=
  \left|
    \sum_{u,v}
    \nu_{\uu|\vv}(u|v)\nu_{\vv}(v)
    \log\frac1{\nu_{\uu|\vv}(u|v)}
    -\sum_{u,v}
    \nu_{\uu|\vv}(u|v)\mu_{V}(v)
    \frac1{\nu_{\uu|\vv}(u|v)}
  \right|
  \notag
  \\
  &\leq
  \sum_{v}
  |\nu_{\vv}(v)-\mu_{V}(v)|
  \sum_{u}
  \nu_{\uu|\vv}(u|v)
  \log\frac1{\nu_{\uu|\vv}(u|v)}
  \notag
  \\
  &=
  \sum_{v}
  |\nu_{\vv}(v)-\mu_{V}(v)|
  H(\nu_{\uu|\vv}(\cdot|v))
  \notag
  \\
  &\leq
  \sqrt{2\gamma}\log|\U|,
  \label{eq:diff-cH}
\end{align}
where the last inequality comes from Lemma \ref{lem:type} and the fact
that $H(\nu_{\uu|\vv}(\cdot|v))\leq \log|\U|$. From (\ref{eq:diff-Hc})
and (\ref{eq:diff-cH}), we have
\begin{align}
  |H(\uu|\vv)-H(U|V)|
  &\leq
  |H(\nu_{\uu|\vv}|\nu_{\vv})-H(\mu_{U|V}|\nu_{\vv})|
  +|H(\mu_{U|V}|\nu_{\vv})-H(U|V)|
  \notag
  \\
  &\leq \iota_{\U|\V}(\gamma'|\gamma).
\end{align}
\end{proof}

\begin{lem}[{\cite[Lemma 26]{HASH}\cite[Theorem 2.8]{UYE}}]
\label{lem:typical-prob}
For any $\gamma>0$, and $\vv\in\V^n$,
\begin{align*}
  \mu_U([\T_{U,\gamma}]^c)
  &\leq
  2^{-n[\gamma-\lambda_{\U}]}
  \\
  \mu_{U|V}([\T_{U|V,\gamma}(\vv)]^c|\vv)
  &\leq
  2^{-n[\gamma-\lambda_{\U\V}]},
\end{align*}
where $\lambda_{\U}$ and $\lambda_{\U\V}$ are defined in (\ref{eq:lambda}).
\end{lem}

\begin{lem}[{\cite[Lemma 27]{HASH}\cite[Theorem 2.9]{UYE}}]
\label{lem:typical-number-bound}
For any $\gamma>0$, $\gamma'>0$, and $\vv\in\T_{V,\gamma}$,
\begin{align*}
  \left|
    \frac 1{n}\log |\T_{U,\gamma}| - H(U)
  \right|
  &\leq
  \eta_{\U}(\gamma)
  \\
  \left|
    \frac 1{n}\log |\T_{U|V,\gamma'}(\vv)| - H(U|V)
  \right|
  &\leq
  \eta_{\U|\V}(\gamma'|\gamma),
\end{align*}
where $\eta_{\U}(\gamma)$ and $\eta_{\U|\V}(\gamma'|\gamma)$ are defined
in (\ref{eq:def-eta}) and (\ref{eq:def-etac}), respectively.
\end{lem}

\section*{Acknowledgements}
We thank Prof.~T.S.~Han for helpful discussions. Constructive comments,
suggestions, and references by anonymous reviewers have significantly
improved the presentation of our results.

\end{document}